\documentclass[12pt]{article}

\usepackage{amsmath,amssymb,amsthm}
\usepackage{geometry}
\usepackage{booktabs}
\usepackage{array}
\usepackage{xcolor}
\usepackage{hyperref}
\usepackage{enumitem}
\usepackage{algorithm}
\usepackage{algpseudocode}
\usepackage{graphicx}

\hypersetup{
  colorlinks = true,
  linkcolor  = blue!70!black,
  urlcolor   = blue!70!black,
  citecolor  = blue!70!black
}

\theoremstyle{plain}
\newtheorem{proposition}{Proposition}[section]

\theoremstyle{definition}
\newtheorem{defn}{Definition}[section]
\newtheorem{remark}{Remark}[section]

\algrenewcommand\algorithmicrequire{\textbf{Input:}}
\algrenewcommand\algorithmicensure{\textbf{Output:}}
\floatname{algorithm}{Algorithm}

\title{Apple-Peel Unfolding in Three and Four Dimensions:\\[0.3em]
       Spiral and Zonal Selection Rules}

\author{
  T.~Yoshino\thanks{Corresponding author.
    Email: \texttt{tyoshino@toyo.jp}}\\
  \small Department of Mechanical Engineering, Toyo University,\\
  \small 2100 Kujirai, Kawagoe, 350-8585, Japan
  \and
  S.~Chaidee\\
  \small Department of Mathematics, Faculty of Science,
         Chiang Mai University,\\
  \small 239 Huay Kaew Road, Muang District,
         Chiang Mai, 50200, Thailand\\
  \small Advanced Research Center for Computational Simulation,
         Chiang Mai University,\\
  \small 239 Huay Kaew Road, Muang District,
         Chiang Mai, 50200, Thailand
}

\date{2026/05/19}

\begin{document}
\maketitle


\bigskip

\begin{abstract}
Apple-Peel Unfolding is a greedy algorithm that selects the faces (or cells)
of a polyhedron (or polytope) one at a time in a spiral order,
producing a net analogous to peeling an apple in a single continuous strip.
We define two face-selection rules---RS (Spiral rule: minimum signed determinant,
i.e.\ sharpest clockwise turn) and
RZ (Zonal rule: maximum coordinate along the peeling axis)---and systematically
evaluate their unfolding success rates on
(i)~the five Platonic solids,
(ii)~the thirteen Archimedean solids, and
(iii)~the six regular convex 4-polytopes.
A principal contribution is a three-way classification of each solid as
\emph{Perfect} (every starting pair yields a complete net),
\emph{Possible} (at least one pair succeeds), or
\emph{Impossible} (no pair succeeds), together with an equivariance
argument showing that face-transitive solids are confined to the
$0/100\%$ dichotomy.
RZ achieves the highest success rates in most cases;
for the regular 4-polytopes it is the only rule yielding non-zero
results for the 120-cell, where it achieves a Perfect result
(1,440/1,440 pairs).
We note that \emph{ordering success} (completing the greedy traversal)
and \emph{geometric validity} (no self-intersection in the 3D realization)
are distinct: every 120-cell ordering produces a self-intersecting 3D net,
so the 120-cell has zero valid 3D nets despite its Perfect ordering result.
The 600-cell is Impossible under all rules tested.
\end{abstract}

\tableofcontents

\section{Introduction}

The problem of constructing a non-overlapping net (unfolding) of a convex
polyhedron by cutting along edges dates back to Albrecht D\"{u}rer (1525)
and remains open in general~\cite{Shephard1975,Demaine2007,Pak2010};
in the non-convex setting, Bern et al.~\cite{Bern2003} exhibited
polyhedra with convex faces that admit no edge-unfolding into a
simple planar net at all.
Algorithmic studies of unfolding heuristics on convex polyhedra go back
at least to Schlickenrieder~\cite{Schlickenrieder1997}, who introduced
and compared steepest-edge, shortest-path, and spiral strategies.
O'Rourke~\cite{ORourke2015} proved that all Platonic and Archimedean solids
admit non-overlapping \emph{spiral} unfoldings defined by a continuous band,
and Lubiw et al.~\cite{Lubiw2010} showed that all such solids admit
\emph{Hamiltonian} (Zipper) unfoldings.
Aronov and O'Rourke~\cite{AronovORourke1992} established that the star
unfolding of any convex polytope produces a non-overlapping net.
While Kaino~\cite{Kaino2019} has explored apple-peel-type foldouts of
4-polytopes---including layer-based studies of the 120-cell and 600-cell,
the latter left unresolved by both methods---and
Akitaya et al.~\cite{Akitaya2024} have studied
path-unfolding of the tesseract (8-cell), those works do not provide a
unified greedy algorithm applicable to all six regular 4-polytopes and
all starting pairs; to our knowledge, no such systematic algorithmic
formulation has been given prior to the present work.

In this paper we introduce \textbf{Apple-Peel Unfolding}, a greedy algorithm
that selects the faces of a polyhedron one at a time in a spiral order,
inspired by the single continuous strip produced when peeling an apple.
When the resulting net is viewed from the outside of the polyhedron,
the peel path winds \emph{clockwise}---the orientation of a right-handed
person peeling in the leftward direction.
Unlike continuous-band or edge-cut approaches, the algorithm is governed
by an explicit \emph{face-selection rule}, and the choice of rule is central
to its performance.
A preliminary version of the 3D algorithm appeared in a companion
preprint~\cite{Yoshino2026arXiv}; the present paper provides a
self-contained treatment of the algorithm, a systematic two-rule
comparison on all Platonic and Archimedean solids, and the first
extension to four-dimensional polytopes. 
Specifically, we define two selection rules (RS, RZ) via the signed
three-dimensional determinant $\det(\mathbf{c}_1,\mathbf{c}_k,\mathbf{c}_j)$
(Section~\ref{sec:algorithm}), apply them to all Platonic and Archimedean
solids (Section~\ref{sec:3d}), extend the algorithm to four dimensions via
a global $\mathbf{c}_1$--$\mathbf{c}_2$ determinant condition
(Section~\ref{sec:4d}; cross-dimensional summary in
Table~\ref{tab:summary}), and illustrate the algorithm's behavior on
representative regular 4-polytopes (Section~\ref{sec:examples}).

\section{Algorithm}
\label{sec:algorithm}

\subsection{Geometric Setup}

Given a polyhedron with vertex coordinates and face list, the algorithm
proceeds as follows.
\begin{enumerate}[nosep]
  \item Translate the vertex set so that its centroid is at the origin.
  \item Choose a starting face $F_1$.
        Rotate so that the centroid $\mathbf{c}_1$ of $F_1$ aligns with
        the $+z$ axis (\emph{Face-up} orientation).
        In later experiments (Section~\ref{sec:3d}),
        we evaluate every face as $F_1$; because the Face-up rotation
        differs for each choice, this is equivalent to running the
        algorithm with the solid tilted into each of its face-up orientations.
  \item Fix a second face $F_2$ adjacent to $F_1$.
  \item Greedily select subsequent faces one at a time until all faces
        are visited or no unvisited adjacent face remains.
\end{enumerate}
An unfolding is \textbf{successful} if every face is visited exactly once.
We evaluate all ordered pairs $(F_1, F_2)$ with $F_2$ adjacent to $F_1$.
Figure~\ref{fig:axes} shows the resulting coordinate system and face labeling.

\begin{figure}[htbp]
\centering
\includegraphics[width=0.5\textwidth]{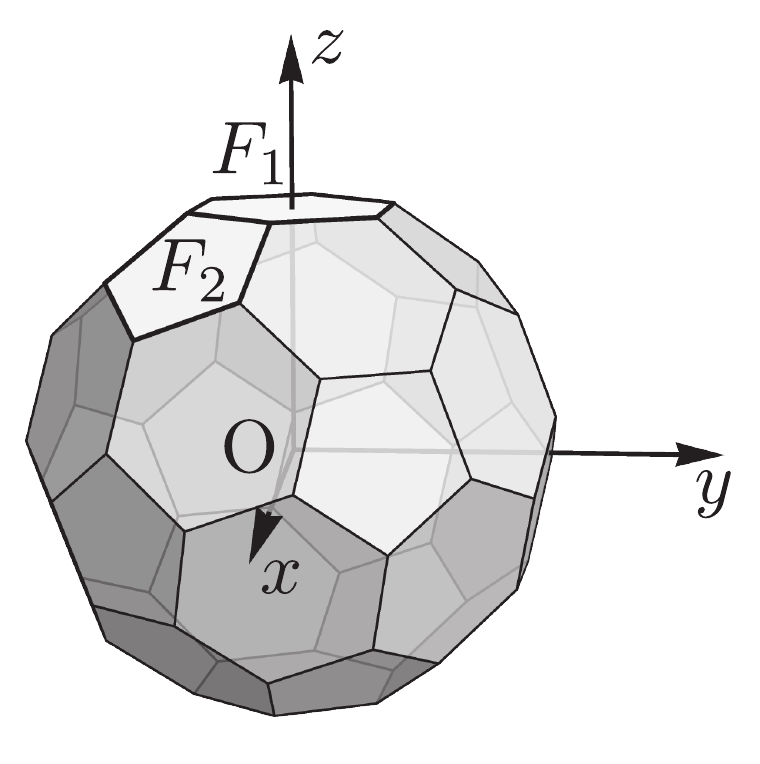}
\caption{Coordinate axes and face labeling used in the Apple-Peel algorithm.
  The polyhedron is rotated so that the starting face $F_1$ lies at the
  ``north pole'' ($+z$ axis).}
\label{fig:axes}
\end{figure}

Let $\mathbf{c}_1$ be the centroid of the starting face $F_1$ (fixed throughout),
$\mathbf{c}_k$ the centroid of the current face $F_k$, and
$\mathbf{c}_j$ the centroid of a candidate adjacent face $F_j$.
We call $F_j$ a \textbf{right candidate} if
\begin{equation}
  \det(\mathbf{c}_1,\,\mathbf{c}_k,\,\mathbf{c}_j) \;\leq\; \varepsilon,
  \qquad \varepsilon = 10^{-10}.
  \label{eq:left}
\end{equation}
The right half-space consists of faces whose centroid lies clockwise
of the current direction when viewed from outside the polyhedron;
selecting from it enforces the CW spiral at every step
(Figure~\ref{fig:leftcond}).
The reference point $\mathbf{c}_1$ is the starting-face centroid and
remains \emph{fixed} at every step of the peeling sequence.
This \textbf{global} reference guarantees that the spiral direction
around the $z$-axis is consistent throughout; using the previous face
centroid $\mathbf{c}_{k-1}$ as a \emph{local} reference allows the
spiral direction to drift between steps and breaks the equivariance of
the algorithm under the symmetry group (see Proposition~\ref{prop:equivariance}).
The threshold $\varepsilon=10^{-10}$ absorbs floating-point noise in the
determinant, which arises for solids with irrational coordinates
(e.g., the Dodecahedron, whose vertices involve the golden ratio)
and can cause $\det\approx 0$ values to flip sign across symmetry-equivalent pairs.

\begin{figure}[htbp]
\centering
\includegraphics[width=0.55\textwidth]{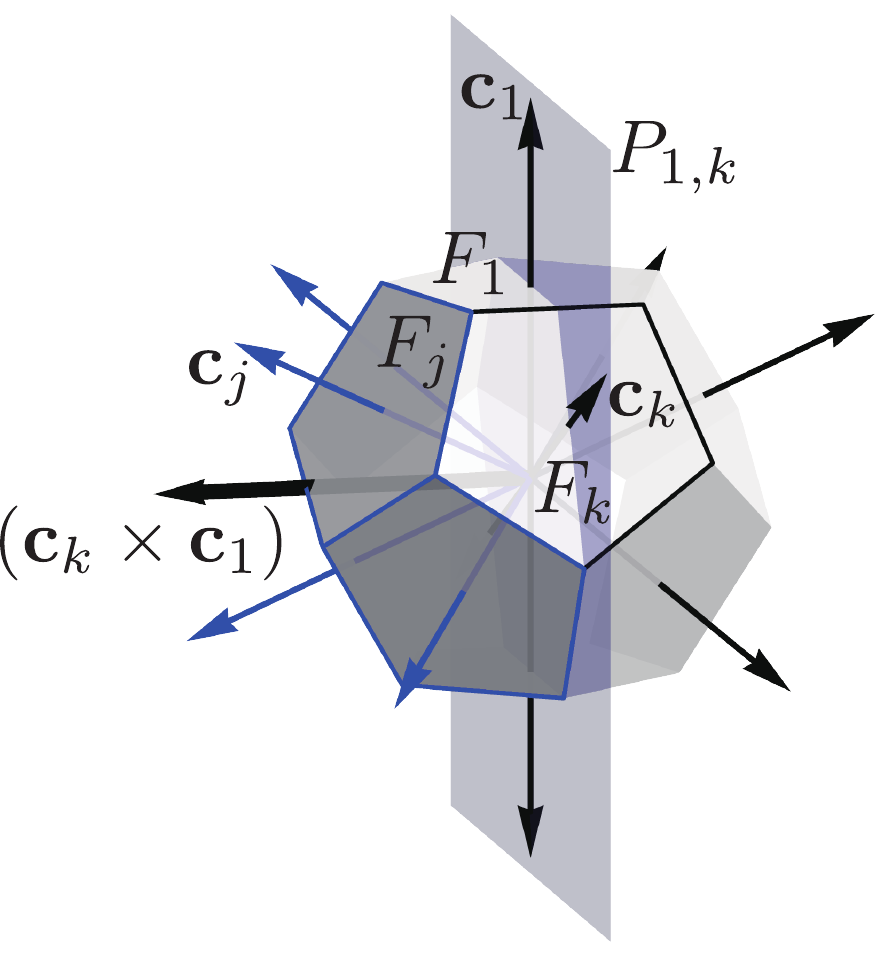}
\caption{Illustration of the right half-space condition~\eqref{eq:left}.
  $P_{1,k}$ denotes the plane through the origin spanned by
  $\mathbf{c}_1$ and $\mathbf{c}_k$, with normal $\mathbf{c}_k\times\mathbf{c}_1$.
  Blue arrows indicate the outward normal vectors of the candidate faces;
  dark gray faces are unvisited faces not yet in the peeling sequence.
  Candidate face $F_j$ is a right candidate if its centroid
  $\mathbf{c}_j$ lies in the half-space clockwise of $\mathbf{c}_k$
  when viewed from outside, i.e.\ on the $(\mathbf{c}_k\times\mathbf{c}_1)$ side of $P_{1,k}$.}
\label{fig:leftcond}
\end{figure}

The Darboux frame constructs a local coordinate system on the sphere
surface at the current face centroid.
It is used \emph{only} in the fallback branch of one of the two
selection rules (when no right candidate exists);
the rules themselves are introduced in Section~\ref{sec:rules}.

\begin{defn}[Darboux frame on the sphere]
\label{def:darboux}
Given consecutive face centroids $\mathbf{c}_{k-1}$ and $\mathbf{c}_k$,
define:
\begin{align}
  \hat{\mathbf{n}} &= \frac{\mathbf{c}_k}{\|\mathbf{c}_k\|},
  \label{eq:nhat}\\
 \hat{\mathbf{f}} &= \frac{
               (\mathbf{c}_k - \mathbf{c}_{k-1})
               - [(\mathbf{c}_k - \mathbf{c}_{k-1})\cdot\hat{\mathbf{n}}]\,\hat{\mathbf{n}}}{\left\|
               (\mathbf{c}_k - \mathbf{c}_{k-1})
               - [(\mathbf{c}_k - \mathbf{c}_{k-1})\cdot\hat{\mathbf{n}}]\,\hat{\mathbf{n}}
             \right\|},
  \label{eq:fhat}\\
  \hat{\mathbf{l}} &= \hat{\mathbf{n}} \times \hat{\mathbf{f}}.
  \label{eq:lhat}
\end{align}
Here $\hat{\mathbf{n}}$ is the sphere normal, $\hat{\mathbf{f}}$ is the forward direction
(projection of the step vector onto the tangent plane at $\mathbf{c}_k$),
and $\hat{\mathbf{l}}$ is the left direction.
\end{defn}

The \emph{signed angle} from the forward direction to a candidate $F_j$ is
\begin{equation}
  \varphi_j = \operatorname{atan2}\!\bigl(
    \tilde{\mathbf{c}}_j \cdot \hat{\mathbf{l}},\;
    \tilde{\mathbf{c}}_j \cdot \hat{\mathbf{f}}
  \bigr),
  \quad
  \tilde{\mathbf{c}}_j
    = \frac{\mathbf{c}_j - (\mathbf{c}_j\cdot\hat{\mathbf{n}})\hat{\mathbf{n}}}
           {\|\mathbf{c}_j - (\mathbf{c}_j\cdot\hat{\mathbf{n}})\hat{\mathbf{n}}\|}.
  \label{eq:phi}
\end{equation}
Here $\operatorname{atan2}(y,x)$ denotes the four-quadrant inverse tangent,
returning values in $(-\pi,\pi]$, so that $\varphi_j>0$ indicates a left turn
and $\varphi_j<0$ a right turn.

\subsection{Two Selection Rules}
\label{sec:rules}

Let $R$ denote the set of right candidates at the current step
and $P$ the full set of unvisited neighbors.
Note that $\det(\mathbf{c}_1,\mathbf{c}_k,\mathbf{c}_j)
= (\mathbf{c}_1\times\mathbf{c}_k)\cdot\mathbf{c}_j$,
so the right-candidate filter \eqref{eq:left} and the selection criterion below
both measure the same quantity: the signed volume spanned by
$\mathbf{c}_1$, $\mathbf{c}_k$, and the candidate.
Using the same global quantity for both filter and selection ensures
that the spiral direction is consistent across all steps.

Before formalising the rules, we comment on the design choice.
Many \emph{a priori} plausible greedy criteria exist---selecting the
nearest unvisited neighbor, minimizing the angular deviation from the
previous edge direction, or following the smallest dihedral angle
across the shared edge---but each of these depends on \emph{local}
quantities (the previous face or the current edge) and therefore
breaks equivariance under the polyhedron's symmetry group: a symmetry
$\sigma\in G$ rotates the local reference frame together with the
face labels, but the resulting frame interacts differently with the
unvisited neighborhood at each step, so the selection at $(\sigma F_1,
\sigma F_2)$ need not coincide with the $\sigma$-image of the
selection at $(F_1,F_2)$ (cf.\ Proposition~\ref{prop:equivariance}).
Rules that use only the \emph{global} peeling axis $+z$ and the
global reference $\mathbf{c}_1$, in contrast, automatically preserve
equivariance.
Within this equivariant family two natural extremes stand out:
maximizing the azimuthal turn per step (a tight spiral) and
maximizing the axial conservation per step (a zonal sweep).
We study these two as representative endpoints; intermediate
strategies (for instance, a convex combination of the two scores)
are easily formulated, but the two-rule comparison already exposes
the qualitative phenomena---perfect coverage, partial coverage, and
structural impossibility---reported in the remainder of the paper.

We define two rules, named the \textbf{Spiral rule} (RS) and the
\textbf{Zonal rule} (RZ).

\begin{description}
  \item[\textbf{RS} (Spiral rule, min\,Det)]
    Select
    $\displaystyle\arg\min_{j\in R}\,\det(\mathbf{c}_1,\mathbf{c}_k,\mathbf{c}_j)$
    (most negative determinant $=$ sharpest clockwise spiral).
    Fallback when $R=\emptyset$:
    $\arg\min_{j\in P}\,\varphi_j$ via the Darboux frame
    (Definition~\ref{def:darboux});
    ties broken by $\arg\min_{j\in P}\,|\det(\mathbf{c}_1,\mathbf{c}_k,\mathbf{c}_j)|$.

  \item[\textbf{RZ} (Zonal rule, max\,$z$)]
    Select $\displaystyle\arg\max_{j\in R}\,(\mathbf{c}_j)_z$
    (highest centroid along the peeling axis);
    $z$-ties within $10^{-10}$ are broken by
    $\arg\min_{j\in R}\,\det(\mathbf{c}_1,\mathbf{c}_k,\mathbf{c}_j)$.
    Fallback when $R=\emptyset$:
    $\arg\min_{j\in P}\,(\mathbf{c}_j)_z$;
    $z$-ties broken by
    $\arg\min_{j\in P}\,\det(\mathbf{c}_1,\mathbf{c}_k,\mathbf{c}_j)$.
\end{description}

\begin{remark}[Rationale for rule names and criteria]
RS produces a tight clockwise coil around the $z$-axis by always taking the
most negative determinant (sharpest clockwise turn);
the path resembles a helix (spiral) descending from north pole to south pole.
RZ sweeps each latitude band in turn before descending to the next,
producing a zonal scan analogous to geographic latitude zones.

Using $\det(\mathbf{c}_1,\mathbf{c}_k,\mathbf{c}_j)$ as the primary RS
criterion---rather than the Darboux angle $\varphi_j$---aligns the selection
with the right-candidate filter (both use the same global reference $\mathbf{c}_1$),
and ensures equivariance under the symmetry group:
$\det(A\mathbf{c}_1,A\mathbf{c}_k,A\mathbf{c}_j)=\det(A)\,\det(\mathbf{c}_1,\mathbf{c}_k,\mathbf{c}_j)=\det(\mathbf{c}_1,\mathbf{c}_k,\mathbf{c}_j)$
for any $A\in\mathrm{SO}(3)$.
The Darboux frame $\varphi_j$ depends on the local step direction
$\mathbf{c}_{k-1}\to\mathbf{c}_k$ and uses a \emph{local} reference that
changes at every step; its use is retained only in the RS fallback branch
as a geometric tiebreaker when no right candidate exists.
For RZ, the $z$-tiebreak was previously also based on $\varphi_j$; it is
here replaced by $\det$, which is equivalent in practice and
formally consistent.
Algorithm~\ref{alg:peel} gives the complete procedure.
\end{remark}

\begin{algorithm}[tp]
\caption{Apple-Peel Unfolding (3D, rule $r \in \{\text{RS, RZ}\}$,
         variant $v \in \{\text{w},\text{n}\}$).
         variant w = with fallback; variant n = no fallback.}
\label{alg:peel}
\begin{algorithmic}[1]
\Require vertex coordinates, faces, adjacency lists,
         starting pair $(F_1,F_2)$, rule $r$, variant $v$
\Ensure face selection order, success flag
\State $\mathit{order} \gets [F_1,\,F_2]$;\;
       $\mathit{prev} \gets F_1$;\; $\mathit{last} \gets F_2$
\While{$\mathit{last}$ has unvisited adjacent faces}
  \State let $P$ = unvisited neighbors of $\mathit{last}$;\;
         $R \gets \{\,j\in P : \det(\mathbf{c}_1,\mathbf{c}_{\mathit{last}},\mathbf{c}_j)\le \varepsilon\,\}$
         \hfill\eqref{eq:left}
  \If{$R \neq \emptyset$}
    \If{$r = \text{RS}$}
      \State $\mathit{next} \gets \arg\min_{j\in R}\,\det(\mathbf{c}_1,\mathbf{c}_{\mathit{last}},\mathbf{c}_j)$
    \ElsIf{$r = \text{RZ}$}
      \State $\mathit{next} \gets \arg\max_{j\in R}\,(\mathbf{c}_j)_z$;\;
             $z$-ties broken by $\arg\min_j\,\det(\mathbf{c}_1,\mathbf{c}_{\mathit{last}},\mathbf{c}_j)$
    \EndIf
  \ElsIf{$v = \text{n}$} \Comment{no-fallback variant}
    \State \Return $(\mathit{order},\,\text{failure})$
  \Else \Comment{with-fallback variant: $R=\emptyset$, $v=\text{w}$}
    \If{$r = \text{RS}$}
      \State compute Darboux frame $(\hat{\mathbf{n}},\hat{\mathbf{f}},\hat{\mathbf{l}})$
             from $\mathbf{c}_{\mathit{prev}},\mathbf{c}_{\mathit{last}}$
             \hfill (Def.~\ref{def:darboux})
      \State $\mathit{next} \gets \arg\min_{j\in P}\,\varphi_j$;\;
             ties by $\arg\min_j\,|\det(\mathbf{c}_1,\mathbf{c}_{\mathit{last}},\mathbf{c}_j)|$
    \ElsIf{$r = \text{RZ}$}
      \State $\mathit{next} \gets \arg\min_{j\in P}\,(\mathbf{c}_j)_z$;\;
             $z$-ties by $\arg\min_j\,\det(\mathbf{c}_1,\mathbf{c}_{\mathit{last}},\mathbf{c}_j)$
    \EndIf
  \EndIf
  \State append $\mathit{next}$ to $\mathit{order}$;\;
         remove $\mathit{next}$ from all adjacency lists;\;
         $\mathit{prev} \gets \mathit{last}$;\;
         $\mathit{last} \gets \mathit{next}$
\EndWhile
\State \Return $\bigl(\mathit{order},\;|\mathit{order}| = |\mathit{faces}|\bigr)$
\end{algorithmic}
\end{algorithm}

\begin{figure}[htbp]
\centering
\includegraphics[width=0.8\textwidth]{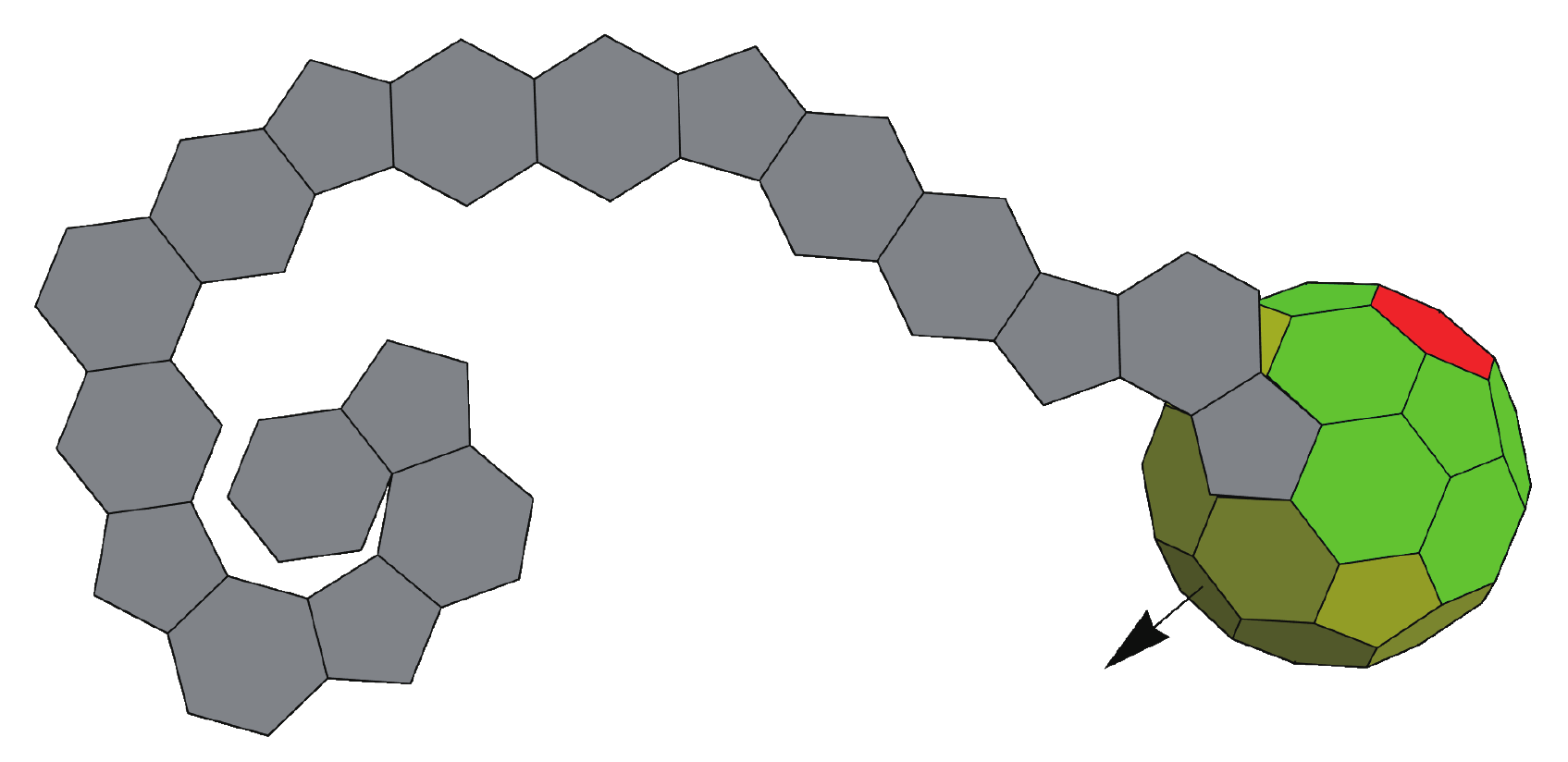}
\caption{Apple-Peel unfolding of the Truncated Icosahedron under RZ (max\,$z$),
  shown after 19 of 32 faces have been selected.
  Gray faces (left) show the unfolded net of the faces peeled so far,
  laid flat by unfolding each consecutive pair of faces along their shared edge.
  On the Truncated Icosahedron (right), olive faces have been peeled and
  green faces remain; the red face is the last face in the sequence,
  which coincides with the antipode of $F_1$ under RZ\@.
  The arrow indicates the direction of $\mathbf{c}_1$.}
\label{fig:peel-example}
\end{figure}

\subsection{Net Construction from Selection Order}
\label{sec:netconstruct}

Given a successful selection order $[F_1,\ldots,F_n]$ produced by
Algorithm~\ref{alg:peel}, a flat 2D net is obtained by sequentially
folding the accumulated strip of faces into the plane of the next face,
using the shared edge as a hinge.
Because $F_k$ ($k\ge2$) is always selected from the unvisited neighbors
of $F_{k-1}$, consecutive faces share exactly one edge.
Figure~\ref{fig:peel-example} shows this process in progress for the
truncated icosahedron after 19 of its 32 faces have been selected.

Let $\mathbf{n}_i$ denote the outward unit normal of face $F_i$.
At step $k$, the faces $F_1,\ldots,F_{k-1}$ (already coplanar after the
previous step) are rotated by the exterior dihedral angle
$\theta_k = \arccos(\mathbf{n}_{k-1}\cdot\mathbf{n}_k)$ around the shared
edge $F_{k-1}\cap F_k$, bringing the strip into $F_k$'s plane.
The rotation axis $\mathbf{a}_k = (\mathbf{n}_{k-1}\times\mathbf{n}_k)/\|\mathbf{n}_{k-1}\times\mathbf{n}_k\|$
is parallel to the shared edge.
After all $n$ steps the entire strip lies in $F_n$'s plane and is
projected to 2D.
Algorithm~\ref{alg:net2d} gives the complete procedure.

\begin{algorithm}[htbp]
\caption{2D Net Construction (3D polyhedra)}
\label{alg:net2d}
\begin{algorithmic}[1]
\Require vertex coordinates $\mathit{vers}$ (3D), face list $\mathit{faces}$,
         selection order $[F_1,\ldots,F_n]$
\Ensure 2D polygon list $\mathit{net} = [\Pi_1,\ldots,\Pi_n]$
\State Compute outward unit normal $\mathbf{n}_i$ for each face $F_i$
       \hfill\Comment{$\mathbf{n}_i \propto (v_2-v_1)\times(v_3-v_1)$, oriented so $\mathbf{n}_i\cdot\bar{v}_i>0$}
\State $S \gets [\,]$ \hfill\Comment{accumulates 3D vertex arrays of placed faces}
\For{$k \gets 2$ \textbf{to} $n$}
  \State Append vertex array of $F_{k-1}$ (original 3D coordinates) to $S$
  \State $\{u,w\} \gets \mathit{faces}[F_{k-1}]\cap \mathit{faces}[F_k]$
         \hfill\Comment{shared edge (hinge)}
  \State $\mathbf{a} \gets \mathrm{Normalize}(\mathbf{n}_{k-1}\times\mathbf{n}_k)$
         \hfill\Comment{rotation axis, parallel to shared edge}
  \State $\theta \gets \arccos\!\bigl(\mathrm{Clip}(\mathbf{n}_{k-1}\cdot\mathbf{n}_k,\,[-1,1])\bigr)$
         \hfill\Comment{exterior dihedral angle}
  \State Rotate every vertex in $S$ by $\theta$ around axis $\mathbf{a}$ through $\mathit{vers}[u]$
         \hfill\Comment{brings strip $F_1\ldots F_{k-1}$ coplanar with $F_k$}
\EndFor
\State Append vertex array of $F_n$ (original 3D coordinates) to $S$
\State Let $R$ be the rotation mapping $\mathbf{n}_n$ to $+z$
\State $\mathbf{p}[v] \gets (R\,v)_{xy}$ for each vertex $v$ in $S$
       \hfill\Comment{project to 2D}
\State Rotate all $\mathbf{p}$ so that $\mathbf{p}[\mathbf{c}_2]-\mathbf{p}[\mathbf{c}_1]$
       aligns with $+x$
       \hfill\Comment{canonical orientation}
\State \Return $\bigl[[\mathbf{p}[v]:v\!\in\! F_k]\;:\;k=1,\ldots,n\bigr]$
\end{algorithmic}
\end{algorithm}

\section{Results: 3D Polyhedra}
\label{sec:3d}

We classify each polyhedron according to the fraction of starting pairs
$(F_1,F_2)$ that yield a successful net under the best available rule:
\emph{Perfect} if every pair succeeds,
\emph{Possible} if at least one pair succeeds, and
\emph{Impossible} if no pair succeeds.
The same classification is used for 4D polytopes in
Section~\ref{sec:4d}.

\subsection{Platonic Solids}

Both RS and RZ achieve 100\% success on all five Platonic solids
with the fallback branch active, classifying each as \emph{Perfect}.
The uniform outcome across all starting pairs is a consequence of
equivariance under the face-transitive rotation group: the success
rate must be either $0\%$ or $100\%$ for each solid
(Proposition~\ref{prop:equivariance}).
Figure~\ref{fig:platonic-nets} shows a representative net for each solid.

\begin{figure}[htbp]
\centering
\includegraphics[width=\textwidth]{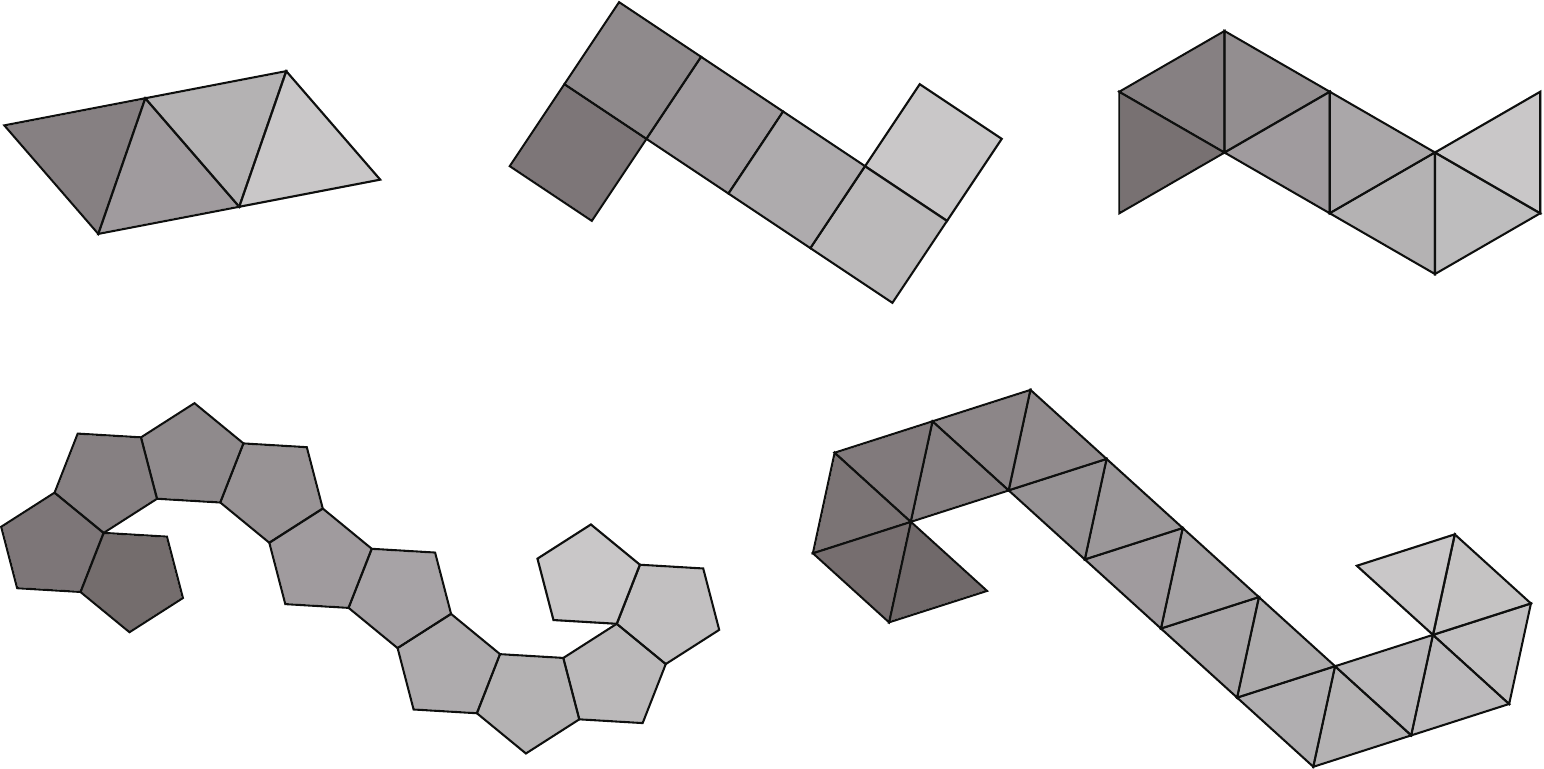}
\caption{Nets of the five Platonic solids obtained by
  Apple-Peel Unfolding (both RS and RZ rules produce identical nets for
  each solid).
  Face shading indicates the selection order:
  dark (first face $F_1$) through light (last face $F_n$).
  By equivariance, all successful orderings of a given solid produce
  congruent nets.}
\label{fig:platonic-nets}
\end{figure}

\begin{proposition}[Equivariance and the $0/100\%$ dichotomy on face-transitive polyhedra]
\label{prop:equivariance}
Let $P$ be a face-transitive convex polyhedron with rotation group $G$.
Then the Apple-Peel algorithm with global right
condition~\eqref{eq:left} and Det-based selection criteria
(rules RS and RZ) is \emph{equivariant} under $G$ in the
following sense: for every $\sigma\in G$ and every ordered adjacent
face pair $(F_1,F_2)$, the cell-selection sequence produced for
$(\sigma F_1,\sigma F_2)$ is the image under $\sigma$ of the
sequence produced for $(F_1,F_2)$.
In particular, the algorithm succeeds for $(F_1,F_2)$ if and only if
it succeeds for $(\sigma F_1,\sigma F_2)$, and since $G$ acts
transitively on ordered adjacent face pairs of a face-transitive
polyhedron, the success rate across all starting pairs is either
$0\%$ or $100\%$.
\end{proposition}

\begin{proof}
Fix $\sigma\in G$ and write
$A = \mathrm{FaceUp}(\sigma F_1)\circ\sigma\circ\mathrm{FaceUp}(F_1)^{-1}$,
so that the Face-up--aligned centroids satisfy
$\mathbf{c}'_j = A\mathbf{c}_j$.
The map $A$ is a proper rotation in $\mathrm{SO}(3)$ (composition of
proper rotations), hence $\det(A)=1$ and
\begin{equation}
  \det(\mathbf{c}'_1,\mathbf{c}'_k,\mathbf{c}'_j)
  = \det(A)\,\det(\mathbf{c}_1,\mathbf{c}_k,\mathbf{c}_j)
  = \det(\mathbf{c}_1,\mathbf{c}_k,\mathbf{c}_j).
  \label{eq:detinvariant}
\end{equation}
Therefore the right-candidate set defined by the global
condition~\eqref{eq:left} is mapped bijectively onto itself by the
permutation that $\sigma$ induces on face labels.

The selection criteria are likewise $A$-invariant.
For RS the primary criterion $\det(\mathbf{c}_1,\mathbf{c}_k,\mathbf{c}_j)$
is invariant by~\eqref{eq:detinvariant}, so its $\arg\min$ commutes
with the action of $A$.
For RZ the primary criterion $z(\mathbf{c}_j)$ is preserved because
$A$ fixes the $+z$ axis.
To see this, trace the image of $+z$ through the three constituent
maps of $A$: the inverse Face-up $\mathrm{FaceUp}(F_1)^{-1}$ sends
$+z$ to $\mathbf{c}_{F_1}$ (the centroid of $F_1$ in the original
unrotated polyhedron), the symmetry $\sigma$ then sends
$\mathbf{c}_{F_1}$ to $\mathbf{c}_{\sigma F_1}$, and finally
$\mathrm{FaceUp}(\sigma F_1)$ sends $\mathbf{c}_{\sigma F_1}$ back to
$+z$.
Hence $A(+z) = +z$, so $A$ is a rotation about the $z$-axis and
preserves the $z$-coordinate of every centroid.
The Det tiebreak is invariant by~\eqref{eq:detinvariant}.
Geometric secondary criteria are used in place of face-index
tie-breaks (see Remark~\ref{rem:symmetry3d} below).

By induction on the step index $k$, the selection sequence transforms
$\sigma$-equivariantly:
$\mathit{order}_{\sigma F_1,\sigma F_2}(k)
= \sigma\bigl(\mathit{order}_{F_1,F_2}(k)\bigr)$
for every $k$.
Consequently the algorithm succeeds for $(F_1,F_2)$ if and only if it
succeeds for $(\sigma F_1,\sigma F_2)$.
Since $G$ acts transitively on ordered adjacent face pairs of a
face-transitive polyhedron, every pair lies in a single orbit and the
success rate is therefore $0\%$ or $100\%$.

A local reference $\mathbf{c}_{k-1}$ would not suffice for this
argument: the composite rotation $A$ depends on $(F_1,\sigma)$ and
changes between symmetry-equivalent steps, so the corresponding
candidate sets would no longer be in bijection under $A$.
\end{proof}

\begin{remark}[Implementation details required for exact equivariance]
\label{rem:symmetry3d}
Two implementation details are required for
Proposition~\ref{prop:equivariance} to hold exactly under
finite-precision arithmetic.
\begin{enumerate}[nosep]
  \item \emph{Epsilon threshold.}
    For the Dodecahedron (golden-ratio coordinates), finite-precision
    arithmetic causes $\det(\mathbf{c}_1,\mathbf{c}_k,\mathbf{c}_j)\approx 0$
    values to fluctuate by $\sim 10^{-16}$.
    The threshold $\varepsilon=10^{-10}$ prevents sign-flips that would
    classify the same candidate as right for one pair but not for a
    symmetry-equivalent pair.
  \item \emph{Geometric tie-breaks.}
    Any ties after the primary criterion are resolved by a geometric
    secondary criterion rather than by face index.
    For RS, the primary criterion is
    $\arg\min_j\,\det(\mathbf{c}_1,\mathbf{c}_k,\mathbf{c}_j)$;
    exact ties use
    $\arg\min_j\,|\det(\mathbf{c}_1,\mathbf{c}_k,\mathbf{c}_j)|$.
    For RZ, $z$-ties are broken by
    $\arg\min_j\,\det(\mathbf{c}_1,\mathbf{c}_k,\mathbf{c}_j)$.
    Both quantities are $\mathrm{SO}(3)$-invariant
    by~\eqref{eq:detinvariant}, preserving equivariance at every level;
    a tie-break by face index would depend on the arbitrary labeling
    of faces and violate equivariance.
    This Det-based convention is the direct 3D analogue of the tie-break
    adopted in the 4D extension (Section~\ref{sec:4dresults}).
    Verification on all five Platonic solids and thirteen Archimedean solids
    confirms that exact Det ties are never triggered in practice;
    the criterion serves as a formal safeguard for completeness.
\end{enumerate}
\end{remark}

\begin{remark}[Strict right condition ($\varepsilon<0$) as a formal variant]
\label{rem:strict-eps}
The standard threshold $\det(\mathbf{c}_1,\mathbf{c}_k,\mathbf{c}_j)\le \varepsilon$
(with $\varepsilon=10^{-10}$, i.e.\ $\varepsilon>0$) admits faces whose
determinant is near zero---those whose centroid lies approximately in the
meridian plane spanned by $\mathbf{c}_1$ and $\mathbf{c}_k$.
Geometrically, such a face lies ``due south'' along the current meridian
and thus neither turns the spiral clockwise nor counterclockwise.

We also tested the \emph{strict} variant $\varepsilon<0$, i.e.\
thresh $= -10^{-10}$, which excludes these boundary faces from
the right-candidate set $R$ and sends them to the fallback branch.
Table~\ref{tab:strict} summarizes the results for the Spiral rule (RS)
and Zonal rule (RZ) on all five Platonic solids.

\begin{table}[ht]
\centering
\caption{Success rates under the strict right condition
  (thresh $=-10^{-10}$, i.e.\ $\varepsilon<0$) for the Spiral rule (RS)
  and Zonal rule (RZ) on all five Platonic solids.
  RS(w)/RZ(w) = with fallback; RS(n)/RZ(n) = without fallback.}
\label{tab:strict}
\renewcommand{\arraystretch}{1.1}
\begin{tabular}{llrrrrr}
\toprule
Solid & $\{p,q\}$ & Pairs & RS(w) & RS(n) & RZ(w) & RZ(n) \\
\midrule
Tetrahedron  & $\{3,3\}$ & 12 & 100\% & 100\% & 100\% & 100\% \\
Cube         & $\{4,3\}$ & 24 & 100\% &   0\% & 100\% &   0\% \\
Octahedron   & $\{3,4\}$ & 24 & 100\% &   0\% & 100\% &   0\% \\
Dodecahedron & $\{5,3\}$ & 60 & 100\% &   0\% & 100\% &   0\% \\
Icosahedron  & $\{3,5\}$ & 60 & 100\% &   0\% & 100\% &   0\% \\
\bottomrule
\end{tabular}
\end{table}

Three observations:
\begin{enumerate}[nosep]
  \item \emph{Equivariance is preserved.}
    All entries are 0\% or 100\%, confirming that the strict condition is
    also equivariant (the determinant inequality is preserved under proper
    rotations regardless of the sign of the threshold).
  \item \emph{With fallback: 100\% for all solids.}
    The ``due south'' faces excluded from $L$ are redirected to the fallback
    branch.  The RS fallback ($\arg\min\varphi_j$) and the RZ fallback
    ($\arg\min z$) both select these faces via a different criterion, so the
    same peeling sequence is ultimately produced.
  \item \emph{Without fallback: 100\% only for the Tetrahedron.}
    The Tetrahedron has no face whose centroid lies exactly in the meridian
    plane after Face-up alignment, so the strict condition causes no
    additional terminations.  For all other Platonic solids there exists at
    least one step where the sole viable candidate has $\det\approx 0$; the
    strict condition empties $L$ at that step and the no-fallback variant
    terminates immediately.
\end{enumerate}
The standard $\varepsilon>0$ setting is therefore preferred: it is the
unique choice that (a) preserves equivariance and (b) maximizes the
success rate without fallback.
\end{remark}

\paragraph{Dodecahedron.}
The Dodecahedron (12 pentagonal faces) achieves 100\% success under
both RS and RZ with the fallback branch active.
Under the strict no-fallback variant ($\varepsilon<0$), the success
rate drops to 0\% for all Platonic solids except the Tetrahedron
(Table~\ref{tab:strict}).

RS takes the sharpest available right turn at each step, wrapping the
path clockwise around the surface and reaching the opposite pole in all
60 starting pairs.

Kaino \cite{Kaino2019} illustrates two net types for the Dodecahedron
(``Dodecahedron~1'' and ``Dodecahedron~2'' in Fig.~1 of that work).
The Apple-Peel algorithm under both RS and RZ produces only
``Dodecahedron~2''; ``Dodecahedron~1'' does not arise under the
right-half-space spiral constraint, indicating that the Apple-Peel
family does not exhaust all net topologies of the Dodecahedron.

\subsection{Archimedean Solids}

The Archimedean solid experiments apply the same Face-up pre-rotation
as the Platonic solid computations: the solid is rotated so that the
centroid of the first face $F_1$ aligns with the $+z$ axis before
running the algorithm.
The equal azimuthal spacing implied by
Proposition~\ref{prop:equivariance} does not apply here because
Archimedean solids lack the full face-transitive symmetry of Platonic
solids; however, the Face-up condition ensures
that comparisons across rules are made on a consistent geometric footing.

\begin{table}[ht]
\centering
\caption{Unfolding success rates (\%) on the 13 Archimedean solids.
  ``w'' = with fallback; ``n'' = no fallback.
  Bold entries mark the best result per row.}
\label{tab:archimedean}
\renewcommand{\arraystretch}{1.2}
\small\setlength{\tabcolsep}{4pt}
\begin{tabular}{llrrrrrr}
\toprule
& & & & \multicolumn{2}{c}{RS\,(\%)} & \multicolumn{2}{c}{RZ\,(\%)} \\
\cmidrule(lr){5-6}\cmidrule(lr){7-8}
Polyhedron & v.c. & $F$ & Pairs & w & n & w & n \\
\midrule
Truncated Tetrahedron        & 3.6.6       &  8 &  36 & \textbf{66.7} & \textbf{66.7} & 33.3 & 33.3 \\
Cuboctahedron                & 3.4.3.4     & 14 &  48 &  0.0 &  0.0 &  0.0 &  0.0 \\
Truncated Cube               & 3.8.8       & 14 &  72 &  0.0 &  0.0 &  0.0 &  0.0 \\
Truncated Octahedron         & 4.6.6       & 14 &  72 & 41.7 & 41.7 & \textbf{100.0} & \textbf{100.0} \\
Rhombicuboctahedron          & 3.4.4.4     & 26 &  96 &  0.0 &  0.0 &  0.0 &  0.0 \\
Truncated Cuboctahedron      & 4.6.8       & 26 & 144 &  0.0 &  0.0 & \textbf{100.0} & \textbf{100.0} \\
Snub Cube                    & 3.3.3.3.4   & 38 & 120 & 20.0 &  0.0 & \textbf{40.0} & 20.0 \\
Icosidodecahedron            & 3.5.3.5     & 32 & 120 &  0.0 &  0.0 &  0.0 &  0.0 \\
Truncated Dodecahedron       & 3.10.10     & 32 & 180 &  0.0 &  0.0 &  0.0 &  0.0 \\
Truncated Icosahedron        & 5.6.6       & 32 & 180 &  0.0 &  0.0 & \textbf{100.0} & \textbf{100.0} \\
Rhombicosidodecahedron       & 3.4.5.4     & 62 & 240 &  0.0 &  0.0 &  0.0 &  0.0 \\
Truncated Icosidodecahedron  & 4.6.10      & 62 & 360 &  0.0 &  0.0 & \textbf{66.7} & \textbf{66.7} \\
Snub Dodecahedron            & 3.3.3.3.5   & 92 & 300 &  0.0 &  0.0 &  0.0 &  0.0 \\
\bottomrule
\end{tabular}
\end{table}

\paragraph{Observations.}
\begin{itemize}
  \item \textbf{Seven of thirteen solids (54\%) yield 0\% under both RS and RZ}:
        cuboctahedron, truncated cube, rhombicuboctahedron, icosidodecahedron,
        truncated dodecahedron, rhombicosidodecahedron, and snub dodecahedron.
        All seven of these solids lack hexagonal faces, in contrast to four
        of the six non-zero solids (truncated octahedron, truncated
        cuboctahedron, truncated icosahedron, truncated icosidodecahedron),
        whose latitude-band structure built on hexagonal faces aligns well
        with the max-$z$ criterion; the truncated tetrahedron is a further
        hexagonal example.
        The snub cube is the only non-hexagonal solid in the non-zero group.
  \item \textbf{RZ achieves 100\% on three solids}:
        truncated octahedron, truncated cuboctahedron, and
        truncated icosahedron; it also achieves the best
        result on the truncated icosidodecahedron (66.7\%).
        RS achieves its highest rate on the truncated tetrahedron (66.7\%)
        and remains positive only on two further solids:
        truncated octahedron (41.7\%) and snub cube (20.0\%).
  \item \textbf{RZ outperforms RS on four solids.}
        Truncated octahedron (RZ\,100\% $>$ RS\,41.7\%),
        truncated cuboctahedron (RZ\,100\% $>$ RS\,0\%),
        truncated icosahedron (RZ\,100\% $>$ RS\,0\%),
        truncated icosidodecahedron (RZ\,66.7\% $>$ RS\,0\%).
        RS outperforms RZ only on the truncated tetrahedron
        (RS\,66.7\% $>$ RZ\,33.3\%).
  \item \textbf{Fallback is essential only for the Snub Cube.}
        Removing the fallback leaves 12 of 13 solids unchanged.
        For the snub cube: RS drops from 20.0\% (24/120) to 0\%,
        and RZ drops from 40.0\% (48/120) to 20.0\% (24/120).
  \item \textbf{Mirror symmetry: counts preserved.}
        For all 13 Archimedean solids, running the algorithm on the mirror
        image (reflection $(x,y,z)\mapsto(-x,y,z)$) yields identical success
        counts for RZ (verified; Section~\ref{sec:mirror}).
        For the chiral snub cube, the successful $(F_1,F_2)$ pairs differ
        between original and mirror image under RZ, revealing sensitivity
        to chirality at the orbit level.
\end{itemize}

\paragraph{Structural interpretation.}
The results across thirteen Archimedean solids support the following
structural principles:
\begin{description}
  \item[Face-type uniformity and RZ dominance]
    The three solids achieving 100\% under RZ (truncated octahedron,
    truncated cuboctahedron, truncated icosahedron) all combine hexagonal
    faces with one or two other regular polygon types; the truncated
    icosidodecahedron reaches 66.7\% under RZ and shares the same
    hexagonal-face structure.
    The latitude-band structure of hexagonal solids is well-suited to
    the max-$z$ criterion, which sweeps each band in turn.
    By contrast, RS (min-Det) takes the sharpest available clockwise turn
    at each step, which can navigate around the hexagonal bands in some
    configurations but follows a more irregular path that not all
    starting pairs can complete successfully.
    The seven solids yielding 0\% under both rules lack hexagonal faces
    entirely, breaking the latitude-band structure that the max-$z$ criterion
    relies on; instead, their face mixtures (triangles combined with
    squares, pentagons, octagons, or decagons depending on the solid)
    yield azimuthal arrangements in which no equivariant greedy rule
    completes the traversal.

  \item[RS outperformance: Truncated Tetrahedron]
    This is the only Archimedean solid where RS (66.7\%) strictly outperforms
    RZ (33.3\%).
    Its 8-face structure (4 triangles + 4 hexagons) with only 36 pairs is
    the smallest Archimedean case; the sharpest-turn criterion navigates
    the alternating face types more reliably than the latitude criterion.

  \item[Path geometry: RZ on the Truncated Icosahedron]
    RZ achieves 100\% and produces a smooth band-by-band traversal through
    consecutive hexagons at the same latitude before descending.
    RS is unsuccessful on this solid with the Det-based criterion.
    A representative RZ sequence is shown in Section~\ref{sec:ex-ti}.
\end{description}
Figure~\ref{fig:arch-all} shows representative nets for four of the six successful Archimedean solids under RZ\@;
the Truncated Icosahedron and Snub Cube are treated separately in
Sections~\ref{sec:ex-ti} and~\ref{sec:ex-snub}.
\begin{figure}[htbp]
\centering
\includegraphics[width=\textwidth]{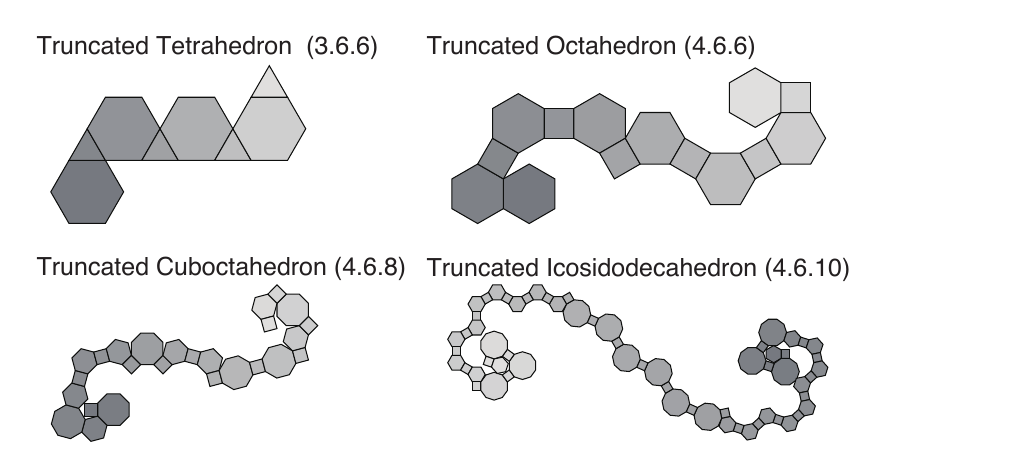}
\caption{Successful nets for four Archimedean solids: Truncated Tetrahedron
  (3.6.6, shown under RZ; RS achieves a higher rate of 66.7\%),
  Truncated Octahedron (4.6.6, RZ), Truncated Cuboctahedron (4.6.8, RZ),
  and Truncated Icosidodecahedron (4.6.10, RZ).
  Solids with all-zero success rates (7 of 13) are omitted;
  the Truncated Icosahedron and Snub Cube are discussed in
  Sections~\ref{sec:ex-ti} and~\ref{sec:ex-snub}.}
\label{fig:arch-all}
\end{figure}

\begin{figure}[htbp]
\centering
\includegraphics[width=0.9\textwidth]{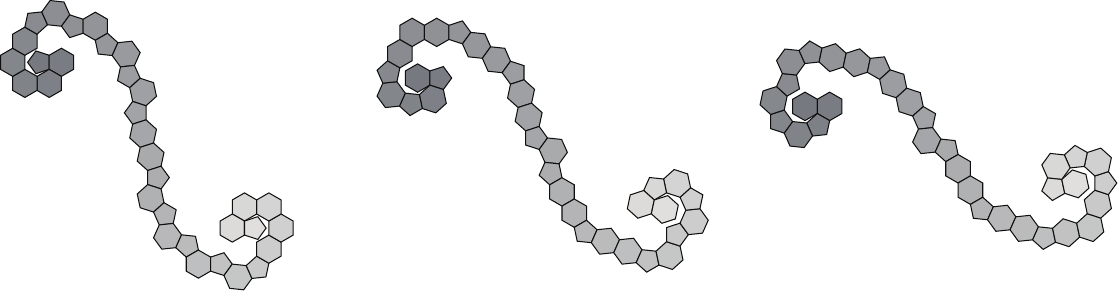}
\caption{Three representative nets of the Truncated Icosahedron under
  the Zonal rule (RZ, max\,$z$).
  Each net traverses the hexagonal faces band by band from north to south.}
\label{fig:ti-nets}
\end{figure}

\subsection{Mirror Symmetry and Chirality Detection}
\label{sec:mirror}

A reflection $\rho\colon (x,y,z)\mapsto(-x,y,z)$ maps each polyhedron to
its mirror image.
Because $\det(\rho)=-1$, the right half-space condition reverses sign:
$\det(\mathbf{c}_1,\mathbf{c}_k,\mathbf{c}_j)\le \varepsilon$ becomes
$-\det(\mathbf{c}_1,\mathbf{c}_k,\mathbf{c}_j)\le \varepsilon$.
Running the right-spiral algorithm on the mirror image is therefore equivalent
to running a \emph{left}-spiral algorithm on the original polyhedron.

\paragraph{Empirical check.}
We ran RS and RZ on both the original and the mirror image of all 13
Archimedean solids. The results are summarized in
Table~\ref{tab:mirror}; for all 13 solids the counts match exactly.

\begin{table}[ht]
\centering
\caption{Success counts for original vs.\ mirror image (all 13 Archimedean
  solids).  RZ counts were verified computationally;
  RS counts are inferred from the equivariance argument for the
  11 amphichiral solids and from the orbit-exchange behavior
  (analogous to RZ; cf.\ Section~\ref{sec:ex-snub})
  for the chiral snub cube.
  Counts are identical in every case.
  Numbers in parentheses indicate the count of geometrically distinct
  net types (equivariance classes of starting pairs);
  $-$ indicates zero successes.}
\label{tab:mirror}
\small
\setlength{\tabcolsep}{5pt}
\begin{tabular}{llrrrrr}
\toprule
Solid & v.c. & Pairs & \multicolumn{2}{c}{RS} & \multicolumn{2}{c}{RZ} \\
\cmidrule(lr){4-5}\cmidrule(lr){6-7}
 & & & Orig & Mirror & Orig & Mirror \\
\midrule
Truncated Tetrahedron       & 3.6.6       &  36 & 24\,(2) & 24\,(2) & 12\,(1) & 12\,(1) \\
Cuboctahedron               & 3.4.3.4     &  48 & \multicolumn{1}{c}{-} & \multicolumn{1}{c}{-} & \multicolumn{1}{c}{-} & \multicolumn{1}{c}{-} \\
Truncated Cube              & 3.8.8       &  72 & \multicolumn{1}{c}{-} & \multicolumn{1}{c}{-} & \multicolumn{1}{c}{-} & \multicolumn{1}{c}{-} \\
Truncated Octahedron        & 4.6.6       &  72 & 30\,(4) & 30\,(4) & 72\,(3) & 72\,(3) \\
Rhombicuboctahedron         & 3.4.4.4     &  96 & \multicolumn{1}{c}{-} & \multicolumn{1}{c}{-} & \multicolumn{1}{c}{-} & \multicolumn{1}{c}{-} \\
Truncated Cuboctahedron     & 4.6.8       & 144 & \multicolumn{1}{c}{-} & \multicolumn{1}{c}{-} &144\,(6) &144\,(6) \\
Snub Cube$^*$               & 3.3.3.3.4   & 120 & 24\,(1) & 24\,(1) & 48\,(2) & 48\,(2) \\
Icosidodecahedron           & 3.5.3.5     & 120 & \multicolumn{1}{c}{-} & \multicolumn{1}{c}{-} & \multicolumn{1}{c}{-} & \multicolumn{1}{c}{-} \\
Truncated Dodecahedron      & 3.10.10     & 180 & \multicolumn{1}{c}{-} & \multicolumn{1}{c}{-} & \multicolumn{1}{c}{-} & \multicolumn{1}{c}{-} \\
Truncated Icosahedron       & 5.6.6       & 180 & \multicolumn{1}{c}{-} & \multicolumn{1}{c}{-} &180\,(3) &180\,(3) \\
Rhombicosidodecahedron      & 3.4.5.4     & 240 & \multicolumn{1}{c}{-} & \multicolumn{1}{c}{-} & \multicolumn{1}{c}{-} & \multicolumn{1}{c}{-} \\
Truncated Icosidodecahedron & 4.6.10      & 360 & \multicolumn{1}{c}{-} & \multicolumn{1}{c}{-} &240\,(4) &240\,(4) \\
Snub Dodecahedron$^*$       & 3.3.3.3.5   & 300 & \multicolumn{1}{c}{-} & \multicolumn{1}{c}{-} & \multicolumn{1}{c}{-} & \multicolumn{1}{c}{-} \\
\bottomrule
\multicolumn{7}{l}{\small$^*$ Chiral solid (enantiomorphic pair).}
\end{tabular}
\end{table}

\paragraph{Amphichiral solids.}
For the 11 amphichiral solids, the mirror image is congruent to the
original under a proper rotation; equivariance then implies identical
counts.

\paragraph{Chiral solids.}
The two chiral solids require separate analysis.
The snub dodecahedron yields 0\% in both cases, so the counts agree
trivially.
For the snub cube, the counts agree but the successful \emph{pairs} differ;
the orbit-exchange structure is examined in Section~\ref{sec:ex-snub}.

\subsection{Truncated Icosahedron: Hexagonal-Band Path under RZ}
\label{sec:ex-ti}

The truncated icosahedron (soccer-ball pattern; 32 faces: 12 pentagons
and 20 hexagons) achieves 100\% success under RZ (all 180 starting pairs),
while RS (min-Det) yields 0\% on this solid.
This contrast illustrates how the two criteria navigate the hexagonal
latitude-band structure differently.

We denote each face type as $\mathtt{P}$ (pentagon) or $\mathtt{H}$
(hexagon).

\medskip
\noindent\textbf{RZ} (max\,$z$; representative starting pair, success 32/32):

\smallskip
\noindent\small
\textit{Types:} \texttt{P H H H H H P H P H H P H H P H H P H P H H P H P H P H P H H P}
\normalsize

The RZ path keeps longer runs of hexagonal faces (mean maximum consecutive
run $\approx 4.1$), traversing several consecutive hexagons at roughly the
same latitude before descending to the next ring---a \textbf{zonal band-by-band traversal}
around the peeling axis.
All 180 starting pairs succeed under RZ, confirming that the
max-$z$ criterion is compatible with the solid's hexagonal
latitude-band geometry.
Figure~\ref{fig:ti-nets} shows a sample of the resulting nets.

\subsection{Snub Cube: Equivariance and Chiral Structure}
\label{sec:ex-snub}

The snub cube (38 faces: 32 triangles and 6 squares) provides an
example where the success rates are partial and differ between rules:
RZ achieves 48 of 120 pairs (40.0\%) while RS achieves 24 of 120 (20.0\%).
The fallback branch is essential for this solid: removing it reduces
RS from 20.0\% to 0\% and RZ from 40.0\% to 20.0\%.

The 32 triangles decompose into two sub-orbits under the chiral
octahedral group $O$ (order 24): 24 \emph{snub} triangles (each sharing
one edge with a square) and 8 \emph{gyrate} triangles (sharing no edge
with a square).
Within each sub-orbit, every pair succeeds or fails uniformly under RZ,
confirming equivariance; the mixed outcome at the face-type level
reflects the distinct geometry of the two sub-orbits, not numerical
inconsistency.
Figure~\ref{fig:snub-nets} shows representative nets for each successful
orbit type under both chiralities:
(a)~Original, $F_1{=}\mathrm{snub}\,\triangle$;
(b)~Original, $F_1{=}\mathrm{gyrate}\,\triangle$;
(c)~Mirror, $F_1{=}\mathrm{snub}\,\triangle$;
(d)~Mirror, $F_1{=}\square$.

\begin{figure}[htbp]
\centering
\includegraphics[width=\textwidth]{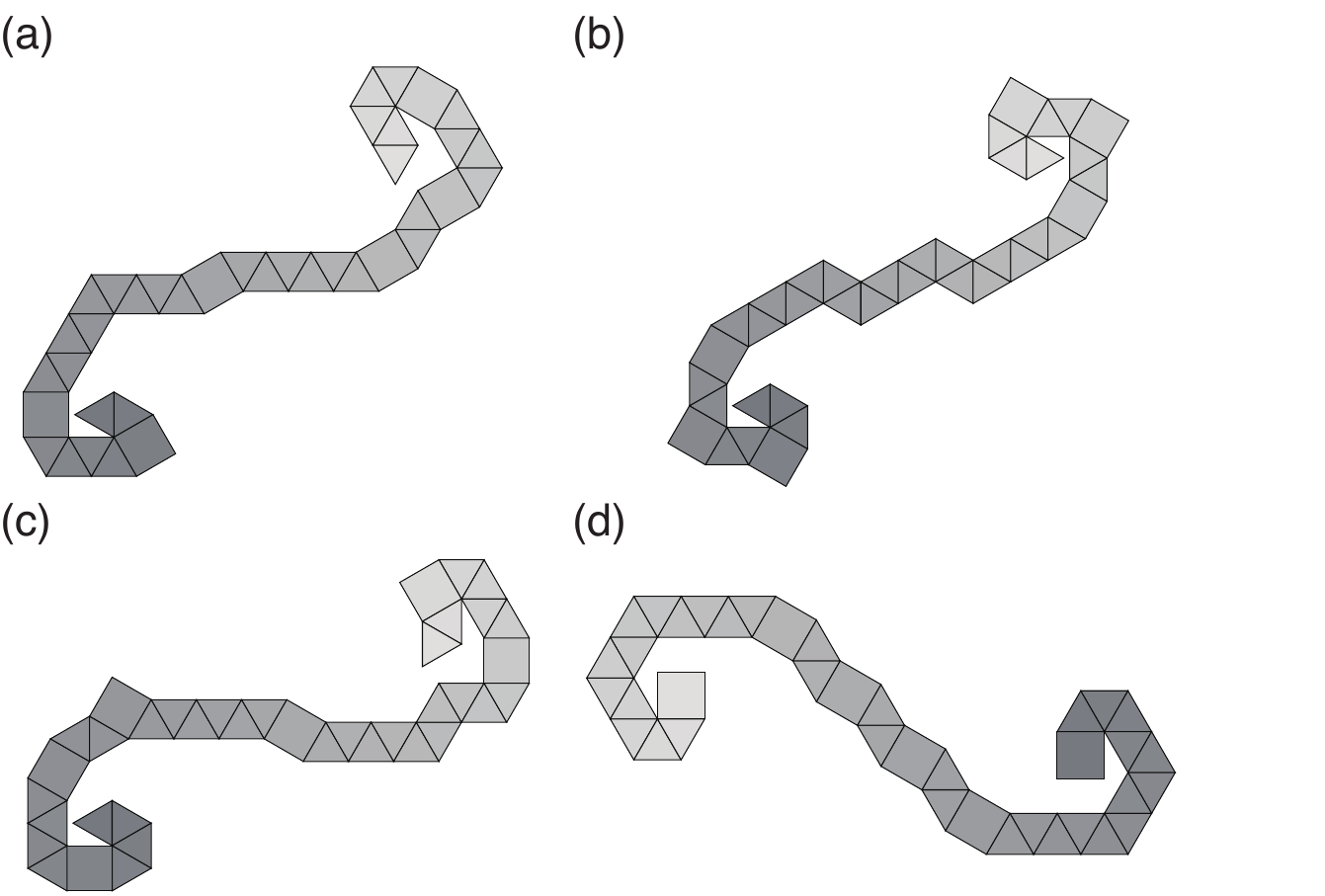}
\caption{Snub Cube (38 faces, RZ rule with fallback): one representative
  net per successful orbit type for each chirality.
  \emph{Original} (left-handed / laevorotatory, 48/120 pairs):
    coordinates taken directly from Mathematica's
    \texttt{PolyhedronData["SnubCube"]} (snub triangles arranged
    counter-clockwise when viewed from outside a square face);
    $F_1{=}\text{snub}\,\triangle$ and $F_1{=}\text{gyrate}\,\triangle$.
  \emph{Mirror} (right-handed / dextrorotatory, 48/120 pairs):
    mirror image of the original (x-coordinates negated);
    $F_1{=}\text{snub}\,\triangle$ and $F_1{=}\square$.
  The square-starting orbit succeeds only for one chirality;
  the two enantiomers share the snub-triangle orbit but have otherwise
  disjoint successful orbit types.
  Color encodes selection order: dark gray (first face)
  $\to$ light gray (last face).}
\label{fig:snub-nets}
\end{figure}

The orbit-exchange behavior under mirror symmetry further confirms that
the algorithm is sensitive to chirality at the orbit level.
Table~\ref{tab:snub-orbits} shows the breakdown by orbit type under RZ
for both the original and the mirror image.

\begin{table}[ht]
\centering
\caption{Success counts by orbit type for the Snub Cube under RZ
  (original and mirror image).}
\label{tab:snub-orbits}
\renewcommand{\arraystretch}{1.15}
\begin{tabular}{lccc}
\toprule
$(F_1\text{-type},\,F_2\text{-type})$ & $n$ & Original & Mirror \\
\midrule
$\{\mathrm{snub},\mathrm{snub}\}$   & 24 & 24/24 & 24/24 \\
$\{\mathrm{gyrate},\mathrm{snub}\}$ & 24 & 24/24 & \ 0/24 \\
$\{\mathrm{sq},\mathrm{snub}\}$     & 24 & \ 0/24 & 24/24 \\
$\{\mathrm{snub},\mathrm{gyrate}\}$ & 24 & \ 0/24 & \ 0/24 \\
$\{\mathrm{snub},\mathrm{sq}\}$     & 24 & \ 0/24 & \ 0/24 \\
\bottomrule
\end{tabular}
\end{table}

The $\{\mathrm{snub},\mathrm{snub}\}$ orbit (24 pairs) succeeds under
both left and right spirals.
The orbit $\{\mathrm{gyrate},\mathrm{snub}\}$ succeeds only under the
original (right spiral), while $\{\mathrm{sq},\mathrm{snub}\}$
succeeds only under the mirror (left spiral).
These two orbits are exchanged by the reflection, reflecting the
geometric fact that gyrate triangles and square faces occupy
antipodally complementary roles in the two enantiomorphs.
The equal orbit sizes (24 each) explain why the total count is identical
despite the different sets of successful pairs.
This behavior confirms that the algorithm correctly distinguishes the two
enantiomorphs of the snub cube, even though the global count is preserved.

\clearpage
\section{Extension to Four Dimensions}
\label{sec:4d}

\subsection{4D Algorithm}

A 4-polytope is represented by vertex coordinates in $\mathbb{R}^4$,
a list of 2-faces (polygons), and a list of 3-cells
(polyhedra given as lists of face indices); for the six regular
4-polytopes we use the standard coordinate sets tabulated by
Coxeter~\cite{Coxeter1973}.
Two cells are \textbf{adjacent} if they share at least one face.

The peeling axis is the $w$-axis.
The algorithm rotates the polytope so that the centroid of the starting
cell $C_1$ aligns with the $+w$ direction; we call this orientation
\textbf{cell-centroid-up}, the 4D analogue of the 3D Face-up rotation.
Despite the name's analogy with Face-up, cell-centroid-up does not align
any face, edge, or vertex of $C_1$ to a preferred direction; it
constrains only the position of $\mathbf{c}_1$ and leaves a residual
rotation about the $w$-axis unresolved.
The same kind of residual exists for 3D Face-up (rotation about $+z$),
but there the algorithm is fully invariant under it; in 4D the residual
is only partially absorbed by the algorithm
(see Remark~\ref{rem:partialequivar4d}).
It is anchored implicitly by the pair $(C_1, C_2)$: the determinant
condition~\eqref{eq:4dleft} uses $(\mathbf{c}_1,\mathbf{c}_2)$ as a
global reference frame.
The algorithm then proceeds as in the 3D case with cells replacing
faces and $w$ replacing $z$.
An alternative orientation, \textbf{3D-face-centroid-up}, aligns
$+w$ with the centroid of the shared 2-face between $C_1$ and a
chosen neighbor; this variant was computed for the 120-cell only,
and the results are presented in Section~\ref{sec:ex-120cell}.
The enumeration of nets of regular convex polytopes in dimension $\le 4$
was studied by Buekenhout and Parker~\cite{Buekenhout1998};
unfoldings and nets of regular polytopes are further explored
in~\cite{Devadoss2022}.

We employ the global $\mathbf{c}_1$--$\mathbf{c}_2$ determinant condition,
the direct 4D analogue of the 3D condition~\eqref{eq:left}:
\begin{equation}
  \det(\mathbf{c}_1,\,\mathbf{c}_2,\,\mathbf{c}_k,\,\mathbf{c}_j) \;\ge\; -\varepsilon,
  \quad \varepsilon = 10^{-10},
  \label{eq:4dleft}
\end{equation}
using a fixed global reference pair $(\mathbf{c}_1, \mathbf{c}_2)$.

\begin{proposition}[$\mathrm{SO}(4)$-invariance of the 4D filter and Det selection criterion]
\label{prop:equivariance4d}
For every $A\in\mathrm{SO}(4)$ and every four cell centroids
$\mathbf{c}_1,\mathbf{c}_2,\mathbf{c}_k,\mathbf{c}_j\in\mathbb{R}^4$,
$$
  \det(A\mathbf{c}_1,A\mathbf{c}_2,A\mathbf{c}_k,A\mathbf{c}_j)
  \;=\; \det(\mathbf{c}_1,\mathbf{c}_2,\mathbf{c}_k,\mathbf{c}_j).
$$
Consequently, both the right-candidate set defined
by~\eqref{eq:4dleft} and the Det-based selection score at $k\ge 3$
are invariant under any rotation $A\in\mathrm{SO}(4)$ that maps the
ordered pair $(C_1,C_2)$ to $(\sigma C_1,\sigma C_2)$ for some
symmetry $\sigma$ of the polytope.
\end{proposition}

\begin{proof}
The first equality is immediate from
$\det(A\mathbf{c}_1,A\mathbf{c}_2,A\mathbf{c}_k,A\mathbf{c}_j)
= \det(A)\,\det(\mathbf{c}_1,\mathbf{c}_2,\mathbf{c}_k,\mathbf{c}_j)$
together with $\det(A)=1$ for $A\in\mathrm{SO}(4)$.
The filter
$\det(\mathbf{c}_1,\mathbf{c}_2,\mathbf{c}_k,\mathbf{c}_j)\ge-\varepsilon$
and the selection score
$\det(\mathbf{c}_1,\mathbf{c}_2,\mathbf{c}_k,\mathbf{c}_j)$ are
therefore unchanged by $A$, and the right-candidate set is mapped
bijectively under the label permutation induced by $\sigma$.
\end{proof}

\begin{remark}[Partial equivariance and the role of the $xy$-plane]
\label{rem:partialequivar4d}
Proposition~\ref{prop:equivariance4d} establishes invariance of the
$4$D Det quantities, but not full equivariance of the algorithm as a
whole.
Two further pieces of the algorithm depend on data that the symmetry
$\sigma$ does \emph{not} automatically preserve: (i) the $+w$
direction, which after cell-centroid-up alignment is fixed only by
the stabiliser of $\mathbf{c}_{C_1}$ in the polytope's symmetry
group; and (ii) the orientation of the $xy$-plane used by the $k=2$
score and by the $|\det|<\varepsilon$ fallback at $k\ge 3$, which is
an arbitrary choice during the cell-centroid-up rotation rather than
a $\sigma$-equivariant quantity.
Algorithm equivariance therefore holds only under symmetries that
simultaneously preserve $+w$ and the $xy$-plane orientation, which
in general is a strict subgroup of $\mathrm{Stab}(C_1,C_2)$.
This is why a $0/100\%$ dichotomy analogous to the 3D case
(Proposition~\ref{prop:equivariance}) does \emph{not} hold for the
4D algorithm; a partial success rate such as the $20/64$ obtained
for the 16-cell (Table~\ref{tab:4dglobal}) is consistent with this
weaker equivariance, as discussed further in
Remark~\ref{rem:symmetry}.
\end{remark}

We test both selection rules in this setting:
\begin{description}[nosep]
  \item[RZ (Zonal rule)] Primary: max $w$; secondary tie-break: max
    $\det(\mathbf{c}_1,\mathbf{c}_2,\mathbf{c}_k,\mathbf{c}_j)$.
    Fallback: min $w$, secondary max $|\det|$.
  \item[RS (Spiral rule)] Primary: max
    $\det(\mathbf{c}_1,\mathbf{c}_2,\mathbf{c}_k,\mathbf{c}_j)$;
    secondary tie-break: max $w$.
    Fallback: max $|\det|$, secondary min $w$.
\end{description}
Note that the 4D RS criterion is $\arg\max\det$, whereas the 3D RS
criterion is $\arg\min\det$ (Section~\ref{sec:rules}).
In both dimensions RS selects the right-candidate that is
\emph{most extreme along the admitted half-space direction}: in 3D
the filter~\eqref{eq:left} keeps candidates with
$\det\le+\varepsilon$, so the most extreme value (sharpest clockwise
turn) is the most \emph{negative} determinant, hence $\arg\min$;
in 4D the filter~\eqref{eq:4dleft} keeps candidates with
$\det\ge-\varepsilon$, so the most extreme value is the most
\emph{positive} determinant, hence $\arg\max$.
The two choices are the same geometric rule expressed under filter
conventions of opposite sign.
At step $k=2$ the repeated row $\mathbf{c}_k=\mathbf{c}_2$ forces
$\det=0$, so the geometric score degenerates to the signed
$xy$-area $\mathbf{c}_{2,x}\mathbf{c}_{j,y}-\mathbf{c}_{2,y}\mathbf{c}_{j,x}$ at that step.
Geometrically, the $xy$-area is the signed area swept on the
equatorial $xy$-plane (the orthogonal complement of $+w$) by the
oriented chord $\mathbf{c}_2\to\mathbf{c}_j$; it is the natural
``azimuthal turn'' indicator when the $4$D volume form vanishes
identically.

The same $xy$-area is used as a fallback score at $k\ge 3$ whenever
every candidate determinant satisfies $|\det|<\varepsilon$.
This case arises when the four points
$\mathbf{c}_1,\mathbf{c}_2,\mathbf{c}_k,\mathbf{c}_j$ are
(numerically) coplanar in $\mathbb{R}^4$ for every $j$, so the
volume-form score cannot discriminate among candidates---a
degeneracy observed for the 8-cell (where every cell centroid lies
on a coordinate axis) and at certain steps of the 120-cell traversal
near the south pole.
Without a fallback, the algorithm would in such steps be forced to
break the tie by face index, which is non-equivariant
(cf.\ Remark~\ref{rem:symmetry3d}).
The $xy$-area continuation is a geometrically meaningful and
$\sigma$-equivariant substitute (under symmetries that preserve the
$xy$-plane orientation, in line with
Remark~\ref{rem:partialequivar4d}), and---being identical to the
$k=2$ score---avoids introducing a third selection criterion.
Algorithm~\ref{alg:peel4d} gives the complete procedure.

\begin{algorithm}[tp]
\caption{Apple-Peel Unfolding (4D, global $\mathbf{c}_1$--$\mathbf{c}_2$ reference,
         rule $r\in\{\text{RS},\text{RZ}\}$,
         variant $v\in\{\text{w},\text{n}\}$).
         Variant w = with fallback; variant n = no fallback.}
\label{alg:peel4d}
\begin{algorithmic}[1]
\Require vertex coordinates, cells, adjacency lists,
         starting pair $(C_1,C_2)$, rule $r$, variant $v$,
         tolerance $\varepsilon=10^{-10}$
\Ensure cell selection order, success flag
\State Rotate (cell-centroid-up) so that $\mathbf{c}_{C_1}$ aligns with $+w$
\State $\mathit{order} \gets [C_1,\,C_2]$;\;
       $\mathit{last} \gets C_2$;\; $k \gets 2$
\While{$\mathit{last}$ has unvisited adjacent cells}
  \State $P \gets$ unvisited neighbors of $\mathit{last}$
  \If{$k = 2$}
    \State $R \gets P$;\quad
           $\mathit{score}(j) \gets \mathbf{c}_{2,x}\mathbf{c}_{j,y}-\mathbf{c}_{2,y}\mathbf{c}_{j,x}$
           \hfill ($xy$ cross product)
  \Else
    \State $R \gets \{\,j\in P : \det(\mathbf{c}_1,\mathbf{c}_2,\mathbf{c}_{\mathit{last}},\mathbf{c}_j)\ge -\varepsilon\,\}$
    \If{$\max_{j\in P}|\det(\mathbf{c}_1,\mathbf{c}_2,\mathbf{c}_{\mathit{last}},\mathbf{c}_j)| < \varepsilon$}
      \State $\mathit{score}(j) \gets \mathbf{c}_{2,x}\mathbf{c}_{j,y}-\mathbf{c}_{2,y}\mathbf{c}_{j,x}$
             \hfill (degenerate: $xy$ cross product)
    \Else
      \State $\mathit{score}(j) \gets \det(\mathbf{c}_1,\mathbf{c}_2,\mathbf{c}_{\mathit{last}},\mathbf{c}_j)$
    \EndIf
  \EndIf
  \If{$R \neq \emptyset$}
    \If{$r = \text{RZ}$}
      \State $\mathit{next} \gets \arg\max_{j\in R}(\mathbf{c}_j)_w$;\;
             $w$-ties broken by $\arg\max_j\,\mathit{score}(j)$
    \ElsIf{$r = \text{RS}$}
      \State $\mathit{next} \gets \arg\max_{j\in R}\,\mathit{score}(j)$;\;
             ties broken by $\arg\max_j\,(\mathbf{c}_j)_w$
    \EndIf
  \ElsIf{$v = \text{n}$} \Comment{no-fallback variant}
    \State \Return $(\mathit{order},\,\text{failure})$
  \Else \Comment{with-fallback variant: $R=\emptyset$, $v=\text{w}$}
    \If{$r = \text{RZ}$}
      \State $\mathit{next} \gets \arg\min_{j\in P}(\mathbf{c}_j)_w$;\;
             $w$-ties by $\arg\max_j\,|\mathit{score}(j)|$
    \ElsIf{$r = \text{RS}$}
      \State $\mathit{next} \gets \arg\max_{j\in P}|\mathit{score}(j)|$;\;
             ties by $\arg\min_j\,(\mathbf{c}_j)_w$
    \EndIf
  \EndIf
  \State append $\mathit{next}$ to $\mathit{order}$;\;
         remove $\mathit{next}$ from all adjacency lists;\;
         $\mathit{last} \gets \mathit{next}$;\; $k \gets k+1$
\EndWhile
\State \Return $\bigl(\mathit{order},\;|\mathit{order}|=|\mathit{cells}|\bigr)$
\end{algorithmic}
\end{algorithm}

\subsection{Results: Regular 4-Polytopes}
\label{sec:4dresults}

\begin{table}[ht]
\centering
\caption{Apple-Peel Unfolding with global $\mathbf{c}_1$--$\mathbf{c}_2$ reference.
  ``RZ'' = successful $(C_1,C_2)$ pairs under the Zonal rule;
  ``RS'' = successful pairs under the Spiral rule.
  Each successful pair produces a distinct cell-selection order;
  ``Class'' refers to RZ.}
\label{tab:4dglobal}
\renewcommand{\arraystretch}{1.2}
\begin{tabular}{lllrrrrl}
\toprule
Polytope & $\{p,q,r\}$ & Cell type & Cells & Total & RZ & RS & Class (RZ) \\
\midrule
5-cell   & $\{3,3,3\}$ & tetrahedron  &   5 &    20 &   20 &     20 & \textbf{Perfect}    \\
8-cell   & $\{4,3,3\}$ & cube         &   8 &    48 &   48 &     48 & \textbf{Perfect}    \\
16-cell  & $\{3,3,4\}$ & tetrahedron  &  16 &    64 &   20 &      0 & \textbf{Possible}   \\
24-cell  & $\{3,4,3\}$ & octahedron   &  24 &   192 &  192 &    192 & \textbf{Perfect}    \\
120-cell & $\{5,3,3\}$ & dodecahedron & 120 & 1{,}440 & 1{,}440 &      0 & \textbf{Perfect}    \\
600-cell & $\{3,3,5\}$ & tetrahedron  & 600 & 2{,}400 &      0 &      0 & \textbf{Impossible} \\
\bottomrule
\end{tabular}
\end{table}

\begin{remark}[Symmetry, orientation-sensitivity, and distinct net count]
\label{rem:symmetry}
For a regular polytope, the automorphism group acts transitively on
ordered adjacent cell pairs $(C_1, C_2)$.
The choice of pair affects the algorithm through two independent channels.
\begin{enumerate}[nosep,label=(\roman*)]
  \item \textbf{Cell-centroid-up orientation} (determined by $C_1$).
        The polytope is rotated so that $\mathbf{c}_{C_1}$ aligns with $+w$,
        fixing the coordinate system for all subsequent steps.
        Different choices of $C_1$ produce different rotations, changing
        the $w$-coordinates and $xy$-projections of every cell centroid.
  \item \textbf{Starting right condition} (determined by $C_2$).
        The right half-space filter uses the fixed reference plane
        spanned by $(\mathbf{c}_{C_1}, \mathbf{c}_{C_2})$
        (equation~\eqref{eq:4dleft}).
        Even when $C_1$ is fixed---and hence the coordinate system is
        unchanged---two different choices of $C_2$ define different
        reference planes, yielding different right-candidate sets at
        each step and potentially diverging selection sequences thereafter.
\end{enumerate}
Rules RS and RZ are sensitive to both channels, since they are defined
relative to the fixed $w$-axis and $xy$-plane.
A partial success rate (e.g.\ RZ achieves 20/64 on the 16-cell;
Table~\ref{tab:4dglobal})
therefore reflects the combined orientation-sensitivity of the rule,
not a combinatorial difficulty intrinsic to some starting pairs.

Regarding the distribution of $C_2$ choices: for a fixed $C_1$,
the number of $C_2$ candidates equals the number of faces of the cell
polyhedron---4 for the tetrahedral cells of the 5-cell, 16-cell, and 600-cell;
6 for the cubic cells of the 8-cell; 8 for the octahedral cells of the 24-cell;
and 12 for the dodecahedral cells of the 120-cell.
Unlike the 3D case---where the $C_p$ symmetry of the regular $p$-gon
face guarantees equal azimuthal spacing of $360^\circ/p$
(Proposition~\ref{prop:equivariance})---the equal-spacing property does not
extend to 4D in general.
After cell-centroid-up alignment of $C_1$, the neighboring cell centroids lie
in the three-dimensional equatorial hyperplane ($w=0$) and project onto
the $xy$-plane in a pattern that depends on the specific geometry of the
cell polyhedron and its embedding; the azimuthal spacing is not
uniformly $360^\circ/n$ for $n$ choices.

For the \textbf{5-cell}, both RS and RZ achieve 20/20.
The automorphism group $\mathfrak{S}_5$ (order 120) acts transitively on
the 20 ordered adjacent pairs, and despite the two channels above,
the deterministic algorithm succeeds for every pair and produces a
single abstract cell-selection sequence up to the relabeling induced
by the automorphism.
Consequently, the 20 listed successes are 20 symmetry-equivalent instances
of one underlying unfolding; the number of geometrically distinct
3D nets is \textbf{one}.
(Verification that all 20 three-dimensional net realizations are congruent
is straightforward from the sorted pairwise-distance spectrum of the
vertex sets.)
\end{remark}

\subsection{Three-Dimensional Realisation of 4D Unfoldings}
\label{sec:3dreal}

Each successful cell-selection order can be realized as a
\emph{three-dimensional net} by embedding the cells of the 4-polytope
in $\mathbb{R}^3$ one at a time in selection order.
The procedure consists of three steps:

\begin{enumerate}[nosep]
  \item \textbf{SVD projection.}
        Each cell's vertices (in $\mathbb{R}^4$) are projected onto their
        best-fit 3-dimensional hyperplane via singular value decomposition,
        yielding an isometric local embedding.
  \item \textbf{Procrustes alignment.}
        The local embedding of each successive cell is aligned to the
        previous cell by minimizing the root-mean-square displacement of
        the shared-face vertices (rotation and translation only, no
        scaling).
        If the new cell and the previous cell end up on the same side of
        the shared face, a mirror reflection through that face is applied
        to place the new cell on the exterior side.
  \item \textbf{Overlap check.}
        All pairs of non-adjacent cells are tested for volumetric overlap
        using the Separating Axis Theorem (SAT) with face normals and
        edge cross-products as candidate axes; bounding-box pre-filtering
        is applied for efficiency.
        A tolerance $\varepsilon = 10^{-6}$ ensures that shared faces are
        not counted as overlaps.
\end{enumerate}

Algorithm~\ref{alg:net3d} formalises this procedure.

\begin{algorithm}[htbp]
\caption{3D Net Construction (4D polytopes)}
\label{alg:net3d}
\begin{algorithmic}[1]
\Require 4D vertex coordinates $\mathit{vers}$ (in $\mathbb{R}^4$),
         cell list (each cell as a list of face indices),
         selection order $[C_1,\ldots,C_m]$
\Ensure 3D vertex positions $\{\mathbf{q}[v]\}$, validity flag
\State \textbf{SVD projection of $C_1$:}
       find the best-fit 3-hyperplane of $C_1$'s vertices via SVD;
       set $\mathbf{q}[v]$ to the projected 3D coordinates for each
       $v\in C_1$
\For{$k \gets 2$ \textbf{to} $m$}
  \State \textbf{SVD projection of $C_k$:}
         project $C_k$'s 4D vertices onto $C_k$'s best-fit 3-hyperplane;
         let $\tilde{\mathbf{q}}[v]$ denote the local 3D coordinates
  \State $S \gets \text{shared face between } C_{k-1} \text{ and } C_k$
         \hfill\Comment{common 2-face}
  \State \textbf{Procrustes alignment:}
         find rotation $R$ and translation $\mathbf{t}$ minimizing
         $\displaystyle\sum_{v\in S}
           \bigl\|\mathbf{q}[v]-(R\,\tilde{\mathbf{q}}[v]+\mathbf{t})\bigr\|^2$
  \State $\mathbf{q}[v] \gets R\,\tilde{\mathbf{q}}[v]+\mathbf{t}$
         for each $v\in C_k$
  \If{centroid of $C_k$ and centroid of $C_{k-1}$ lie on the \emph{same}
      side of face $S$}
    \State mirror-reflect all $\mathbf{q}[v]$ ($v\in C_k$) through the
           affine hyperplane of $S$
           \hfill\Comment{place $C_k$ on the exterior side}
  \EndIf
\EndFor
\State \textbf{Overlap check:}
       for each non-adjacent pair $(C_i,C_j)$, apply SAT with face
       normals and edge cross-products as separating-axis candidates
       (bounding-box pre-filter; tolerance $\varepsilon_{\mathrm{tol}}=10^{-6}$)
\State \Return $\bigl(\{\mathbf{q}[v]\},\;\text{all non-adjacent pairs separated}\bigr)$
\end{algorithmic}
\end{algorithm}

A realization with no overlapping cells is called a
\textbf{valid 3D net}.
Table~\ref{tab:3dnet} summarizes the results for all successful
orderings produced by RZ (the best-performing rule).

\begin{table}[ht]
\centering
\caption{Three-dimensional net validity for the successful RZ orderings.
  ``Valid'' = no overlapping cell pairs under SAT\@.
  Dashes (---) for the 600-cell indicate no successful orderings
  (Impossible; see Table~\ref{tab:4dglobal}).}
\label{tab:3dnet}
\renewcommand{\arraystretch}{1.2}
\begin{tabular}{lrrrr}
\toprule
Polytope & Orderings & Valid & Overlap & Valid\,(\%) \\
\midrule
5-cell   &   20  &   20  &     0  & 100.0 \\
8-cell   &   48  &   48  &     0  & 100.0 \\
16-cell  &   20  &   20  &     0  & 100.0 \\
24-cell  &  192  &  192  &     0  & 100.0 \\
120-cell & 1440  &    0  &  1440  &   0.0 \\
600-cell &   —   &   —   &    —   &    —  \\
\bottomrule
\end{tabular}
\end{table}

\paragraph{Observations.}
\begin{itemize}[nosep]
  \item The 5-cell, 8-cell, 16-cell, and 24-cell yield \textbf{100\% valid nets}:
        every successful Apple-Peel ordering embeds without overlap in
        $\mathbb{R}^3$.
  \item The 120-cell produces \textbf{zero valid nets} despite yielding
        1440 successful 4D orderings.
        The long spiral path through 120 dodecahedral cells accumulates
        enough curvature along the peel direction that adjacent latitude
        bands collide when laid out in $\mathbb{R}^3$, so self-intersection
        is unavoidable under the greedy strategy
        (see Section~\ref{sec:ex-120cell} for an empirical decomposition
        of overlap pairs by band membership).
  \item The sharp transition from perfect 3D validity (5-, 8-, 16-, 24-cell)
        to zero validity (120-cell) suggests a combinatorial threshold
        related to the number of cells visited along the peel path and the
        angular defect accumulated over the traversal.
\end{itemize}
Figure~\ref{fig:4d-nets} shows the resulting 3D nets for the four polytopes
with valid realizations.
Table~\ref{tab:summary} then consolidates all results across both dimensions
into a single best-rule classification, providing the cross-dimensional
reference point used in the discussion of structural patterns
(Section~\ref{sec:discussion}).

\begin{figure}[htbp]
\centering
\includegraphics[width=\textwidth]{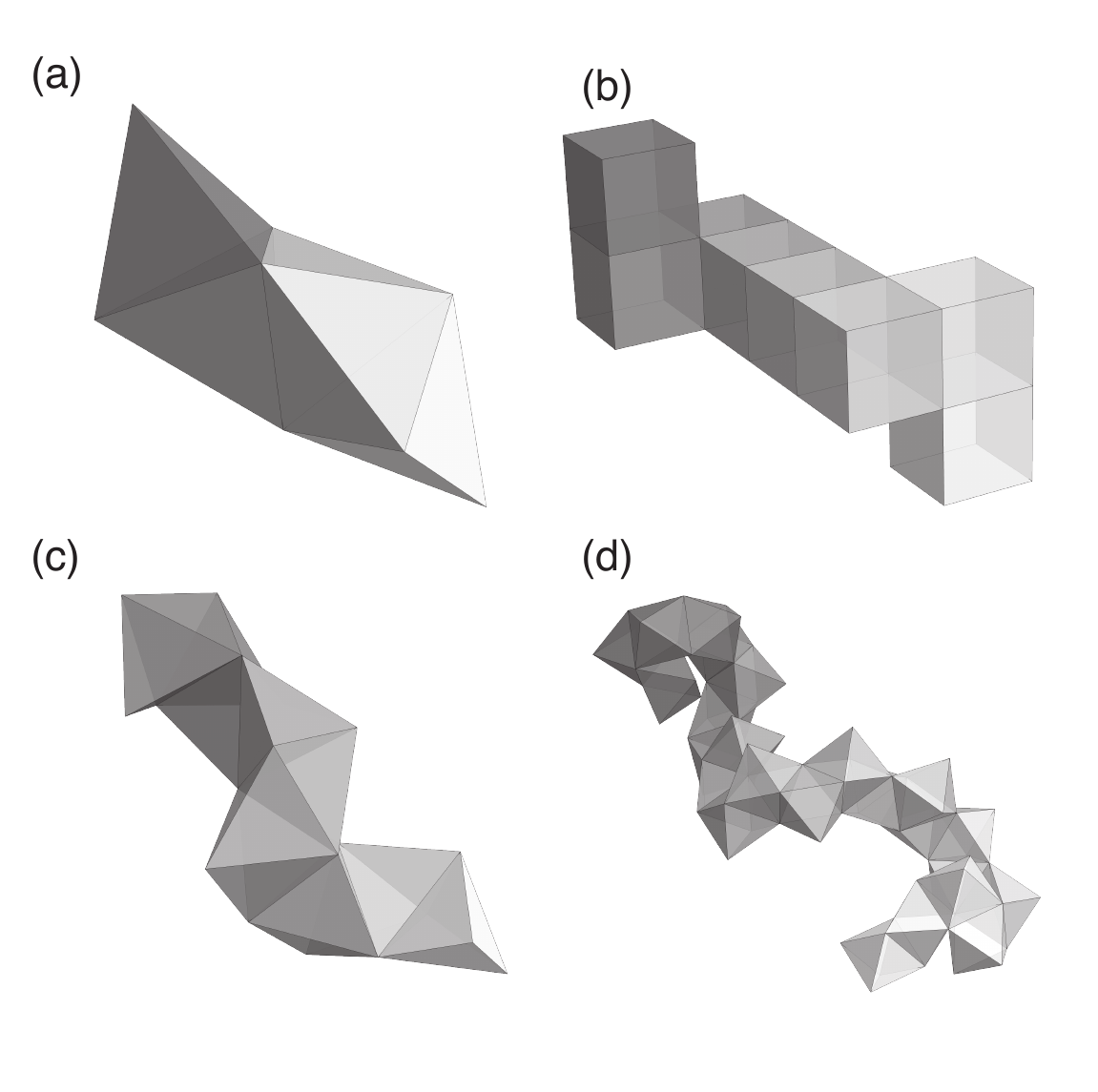}
\caption{Three-dimensional nets under Apple-Peel unfolding (RZ rule):
  (a)~5-cell, (b)~8-cell, (c)~16-cell, (d)~24-cell.
  For each polytope the net geometry is unique up to congruence:
  for the 5-cell ($20/20$), 8-cell ($48/48$, geo-uniqueness verified
  by sorted pairwise vertex-distance spectra), and 24-cell
  ($192/192$), every successful ordering yields the same 3D shape;
  for the 16-cell, $20$ of the $64$ ordered pairs succeed and the
  $20$ resulting realizations are likewise all congruent to the
  shape shown.
  Each net has no cell overlap (Table~\ref{tab:3dnet}).
  Color encodes selection order: dark gray (first cell)
  $\to$ light gray (last cell).}
\label{fig:4d-nets}
\end{figure}

\begin{table}[ht]
\centering
\caption{Classification by best rule across all objects studied.
  4D results use the global $\mathbf{c}_1$--$\mathbf{c}_2$ reference;
  ``Best\,\%'' shows the success rate of the better rule for each
  solid (RZ in all cases except the truncated tetrahedron, where
  RS is better);
  see Table~\ref{tab:4dglobal} for the RS vs.\ RZ comparison.}
\label{tab:summary}
\renewcommand{\arraystretch}{1.2}
\begin{tabular}{lrrl}
\toprule
Object & Best rule & Best\,\% & Class \\
\midrule
\multicolumn{4}{l}{\textit{3D Platonic solids}} \\
Tetrahedron  & any & 100.0 & Perfect \\
Cube         & any & 100.0 & Perfect \\
Octahedron   & any & 100.0 & Perfect \\
Dodecahedron & any & 100.0 & Perfect \\
Icosahedron  & any & 100.0 & Perfect \\
\multicolumn{4}{l}{\textit{3D Archimedean solids (vertex config.)}} \\
Truncated Octahedron (4.6.6)         & RZ & 100.0 & Perfect \\
Truncated Icosahedron (5.6.6)        & RZ & 100.0 & Perfect \\
Truncated Cuboctahedron (4.6.8)      & RZ & 100.0 & Perfect \\
Truncated Icosidodecahedron (4.6.10) & RZ &  66.7 & Possible \\
Snub Cube (3.3.3.3.4)                & RZ &  40.0 & Possible \\
Truncated Tetrahedron (3.6.6)        & RS &  66.7 & Possible \\
7 remaining solids                   & -- &   0.0 & Impossible \\
\multicolumn{4}{l}{\textit{4D regular polytopes (global $\mathbf{c}_1$--$\mathbf{c}_2$ reference)}} \\
5-cell   & RZ & 100.0 & Perfect \\
8-cell   & RZ & 100.0 & Perfect \\
16-cell  & RZ &  31.3 & Possible \\
24-cell  & RZ & 100.0 & Perfect \\
120-cell & RZ & 100.0 & Perfect \\
600-cell & --  &   0.0 & Impossible \\
\bottomrule
\end{tabular}
\end{table}

\clearpage
\section{Computational Examples}
\label{sec:examples}

We illustrate the algorithm's behavior on three representative
regular 4-polytopes.

\subsection{5-Cell and 16-Cell: rule agreement and partial coverage}
\label{sec:ex-5-16cell}

The 5-cell achieves a Perfect result (20/20 pairs) under both rules,
and---by the symmetry argument of Remark~\ref{rem:symmetry}---the
20 successful orderings are symmetry-equivalent and produce a single
3D-congruent net.
Kaino~\cite{Kaino2019} reports the same equivalence for the 5-, 8-,
16-, and 24-cells under a related spiral method.

The 16-cell is the converse extreme: although it shares the
4-neighbor cell graph of the 5-cell, RZ succeeds on only 20 of 64
pairs (31.3\%) and RS is Impossible (0/64), classifying the 16-cell
as Possible rather than Perfect (see Table~\ref{tab:4dglobal}).
On a representative failing pair RZ stalls one cell short of
completion, the remaining cell being non-adjacent to the last chosen
cell---the same dead-end mechanism analyzed in detail for the
600-cell in Section~\ref{sec:discussion}, here without the icosahedral
symmetry that makes the 600-cell uniformly Impossible.

\subsection{120-Cell: Perfect Result under RZ}
\label{sec:ex-120cell}

The 120-cell is Perfect under RZ and Impossible under RS
(Table~\ref{tab:4dglobal}).
Although the RZ primary criterion is max-$w$, the actual traversal
is not a monotone band-by-band descent: the algorithm zigzags
between adjacent latitude bands throughout the 119-step sequence,
and the high cell-connectivity (each dodecahedral cell has 12
neighbors) guarantees either a non-empty right-candidate set or a
viable fallback (min-$w$) bridge at every step---precisely the
combinatorial property the 4-regular 600-cell graph lacks
(Section~\ref{sec:4dresults}).

\paragraph{Latitude-band structure.}
After cell-centroid-up rotation, the 120 cells arrange into
\textbf{9 bands} (grouped by symmetry orbit) numbered from the south
pole ($w$ smallest) to the north pole ($w$ largest, where $C_1$ resides):
\[
  \underbrace{1}_{\text{band 1}} + 12 + 20 + 12 + 30 + 12 + 20 +
  \underbrace{12}_{\text{band 8}} +
  \underbrace{1}_{\text{band 9}} = 120,
\]
with the 12 cells of $C_2$ candidates all residing in band~8
(directly adjacent to $C_1$ in band~9),
confirming that max-$w$ selection at step~$k=2$
is always a tie broken by the $xy$-cross-product score.
The 30-cell equatorial band ($w=0$) is the largest; the fallback
(min-$w$) is essential for bridging this band when the right-half-space
condition excludes all forward candidates.

\paragraph{Geometric diversity of RZ nets.}
Although all 1,440 cell-centroid-up orderings succeed, they do not
produce 1,440 distinct net shapes.
For a fixed starting cell all 12 adjacent $C_2$ choices yield a
successful ordering, but geo-uniqueness testing
(pairwise centroid distances of the 3D embedding, rounded to $10^{-3}$)
reveals only \textbf{7 geometrically distinct net shapes}: two orbits
(sizes 4 and 3) and five singletons.

All 7 nets share a \emph{universal inner core}: the 3D distances
$r_{3D}$ of cell centroids in peeling order are identical for positions
$k=1,\ldots,13$ (the starting cell plus its 12 direct neighbors),
namely $r_{3D}\approx(0,\,1.70,\,1.79,\,1.80,\,1.81,\,1.82,\,1.81,
\,1.96,\,1.99,\,2.04,\,2.06,\,2.02,\,2.09)$.
The nets diverge only from $k=14$ onward.
Computing cylindrical coordinates $(r_{xy},z,\theta)$ of cell centroids
in the unfolded 3D net and summing the incremental turning angle reveals
three spiral-pattern types:
\begin{description}[nosep]
  \item[\textbf{Type A (CCW ascending):}] 5 of 7 nets wind counter-clockwise
    (positive total winding $+0.75$ to $+1.08$ turns), with $z$ increasing
    overall.  This is consistent with RZ's max-$xy$-cross-product
    tie-breaker, which preferentially selects cells in the
    counter-clockwise direction.
    Comprises the orbit of size~4 and four of the five singletons.
  \item[\textbf{Type B (CW / reversed):}] The remaining singleton winds clockwise
    ($-1.14$ turns, $z$: $-7$ to $+4$).
  \item[\textbf{Type C (near-zero / columnar):}] The orbit of size~3 has nearly
    zero net winding ($-0.07$ turns) with $z$ spanning $-4$ to $+23$,
    producing a column-like rather than spiral structure.
\end{description}
The six antipodal $C_2$ pairs (each sharing opposite poles of the
$C_1$ icosahedron) never share a geometric orbit, consistent with
the $k=2$ tie-breaker reversing sign for antipodal choices.
Figure~\ref{fig:120cell-unfold} shows two representative nets contrasting the dominant CCW spiral and the columnar structure.

\begin{figure}[htbp]
\centering
\includegraphics[width=0.85\textwidth]{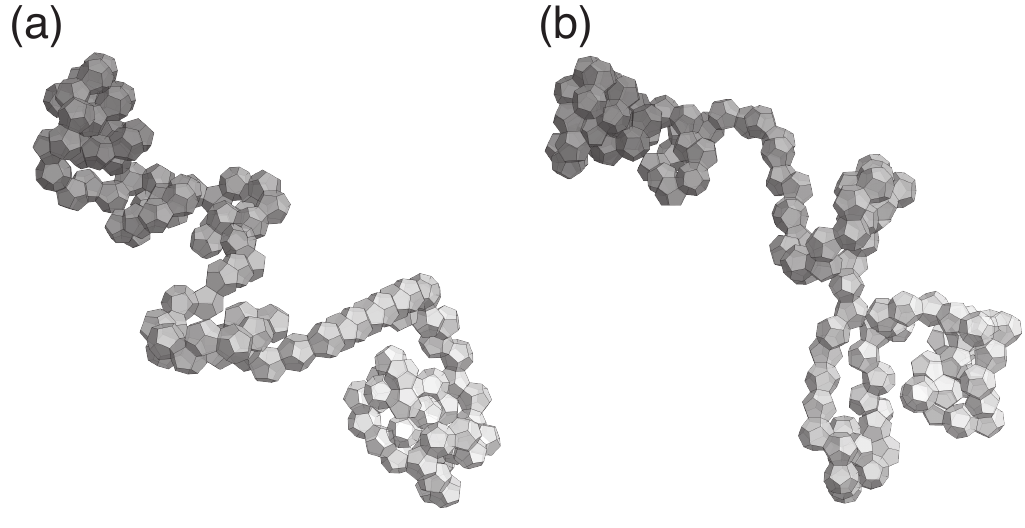}
\caption{Two representative three-dimensional nets of the 120-cell
  under RZ with fallback (120/120 cells), chosen to
  contrast the two most distinct spiral-pattern types among the
  7 geo-unique net shapes.
  \emph{Left} (a): Type~A --- CCW spiral ($+1.07$ turns,
  $z\in[-3,+26]$), the dominant pattern (5 of 7 geo-unique nets).
  \emph{Right} (b): Type~C --- near-zero winding ($-0.07$ turns,
  $z\in[-4,+23]$), a qualitatively different columnar structure.
  Color encodes selection order: dark gray (first cell)
  $\to$ light gray (last cell).
  Both realizations self-intersect in 3D space; no non-overlapping
  net exists for the 120-cell under this algorithm.}
\label{fig:120cell-unfold}
\end{figure}

\paragraph{Structure of self-intersections.}
To understand why all 120-cell nets self-intersect, we examined
whether overlapping cell pairs belong to the same latitude band
or to different bands.
For a fixed starting cell $C_1$, all 12 adjacent $C_2$ choices succeed,
yielding 12 cell-selection orderings;
across these 12 orderings (308 total overlap pairs),
71.4\% of overlaps occur between cells in \emph{different} bands,
while only 28.6\% occur within the same band.
Moreover, the cross-band overlaps are almost exclusively between
\emph{adjacent} bands ($k \leftrightarrow k{+}1$); overlaps between
non-adjacent bands account for only a handful of cases.
This indicates that the self-intersections arise at band boundaries:
as each successive latitude band is unfolded into 3D space,
its dodecahedral cells protrude slightly into the region already
occupied by the preceding band, rather than distant bands
colliding directly.

\paragraph{3D-face-centroid-up orientation.}
We tested whether replacing cell-centroid-up with
\emph{3D-face-centroid-up}---aligning $+w$ with the centroid of
the shared pentagon between $C_1$ and a neighbor, the direct
4D analogue of the 3D Face-up rotation---reduces overlaps.
For a fixed $C_1$, 7 of the 12 adjacent $C_2$ choices complete with
an average of 18.6 overlapping pairs
(vs.\ 24.6 under cell-centroid-up for the same 7 pairs, a $\approx$24\% reduction);
the remaining 5 $C_2$ choices fail to complete entirely.
No valid non-overlapping 3D net is obtained under either orientation,
confirming that the self-intersection of 120-cell nets is structural
rather than an artifact of the initial alignment.
Kaino~\cite{Kaino2019} also studies the 120-cell, but treats it by
examining its layer structure separately rather than applying a greedy
traversal, so the two approaches are not directly comparable.

\subsection{600-Cell: Impossibility and Dead-End Structure}
\label{sec:ex-600cell}

The 600-cell is Impossible under both RZ and RS.
Despite having only 4 neighbors per cell (the same as the 16-cell, which
is Possible under RZ), the 600 tetrahedral cells create a combinatorial
structure that no greedy rule can traverse completely.
Analysis of the termination-step distribution reveals that all 2,400
evaluated $(C_1, C_2)$ pairs halt at one of exactly five step values
(146, 150, 276, 279, or 284), indicating structural bottlenecks related
to the 600-cell's icosahedral symmetry.
We further tested one-step backtracking, 3D-face-centroid-up orientation,
and vertex-up orientation; all variants remain Impossible,
confirming that the obstacle is structural rather than an artifact
of the specific greedy heuristic or orientation choice.

\paragraph{Worked example of a dead-end.}
The mechanism that produces these failures is not exclusion by the
right half-space condition~\eqref{eq:4dleft}, but exhaustion of
\emph{unvisited} neighbors at a cell whose four-cell neighborhood has
already been consumed by earlier RZ steps.
Concretely, for a representative starting pair the algorithm runs
without incident for the first 145 steps and then arrives at step
$k = 146$, at which point all four neighbors of the selected cell
have already appeared in the order.
The four neighbors and their state are summarized in
Table~\ref{tab:600stuck-example}: three sit at a higher $w$-level
and were selected earlier by the RZ max-$w$ greedy rule,
while the unique lower-$w$ exit had already been entered
via a different branch of the spiral.
The Det filter is irrelevant at this step---it operates only on
\emph{unvisited} candidates, of which there are none.

\begin{table}[ht]
\centering
\caption{Per-neighbor state at the stuck step $k=146$ for a
  representative starting pair of the 600-cell under rule RZ
  (with fallback).
  All four neighbors are already in the selection order, leaving the
  set of unvisited candidates empty.  The Det column shows
  $\det(\mathbf{c}_{C_1},\mathbf{c}_{C_2},\mathbf{c}_{\mathrm{last}},\mathbf{c}_j)$
  for completeness; its sign is irrelevant since no unvisited
  candidate remains.}
\label{tab:600stuck-example}
\renewcommand{\arraystretch}{1.15}
\begin{tabular}{llrl}
\toprule
neighbor & $w$-level & $\det(\mathbf{c}_{C_1},\mathbf{c}_{C_2},\mathbf{c}_{\mathrm{last}},\mathbf{c}_j)$ & status \\
\midrule
upper 1  & high ($> w_{\mathrm{last}}$) & $-0.454$    & visited (earlier in order) \\
upper 2  & high ($> w_{\mathrm{last}}$) & $+0.454$    & visited (earlier in order) \\
upper 3  & high ($> w_{\mathrm{last}}$) & $\approx 0$ & visited (earlier in order) \\
lower    & low  ($< w_{\mathrm{last}}$) & $\approx 0$ & visited (earlier in order) \\
\bottomrule
\end{tabular}
\end{table}

The same dead-end mechanism repeats for the remaining $C_2$
choices adjacent to the same $C_1$: some terminate at $k=284$ with
analogous fully-visited neighborhoods, others at $k=146$ at a
different cell with the same pattern.
Empirically, all 2,400 starting pairs terminate at one of exactly
five steps, and the 2,400 pairs partition into orbits under the
rotation group of the 600-cell whose sizes ($742$, $825$, $404$,
$323$, $106$) match the termination-step counts shown earlier; the
4D analogue of the equivariance argument of
Proposition~\ref{prop:equivariance} (see
Proposition~\ref{prop:equivariance4d} above) accounts for the
constancy of the termination step within each orbit, but the
specific partition into five orbits---and, in particular, the values
of the five termination steps---is established here by direct
computation rather than by a structural argument.
The structural reason why these dead-ends are unavoidable is analyzed
in Section~\ref{sec:discussion}.

\clearpage
\section{Discussion}
\label{sec:discussion}

Across both dimensions, high success rates correlate with
\textbf{uniform face types}.
Platonic solids and the 4D analogues (5-cell, 8-cell) with identical
cells perform best.
Solids with mixed face types---square-containing Archimedean solids
and the 600-cell---show low or zero rates.
The 120-cell is an exception: despite its complex dodecahedral cells it
achieves a Perfect result under RZ (100\%), owing to the high connectivity
(12 neighbors per cell) that keeps the determinant-based filter from
blocking all candidates.

\subsection{Spiral vs.\ Zonal Strategy in Four Dimensions}

A striking feature of the 4D results is the asymmetry between RS and RZ.
While both rules agree on the 5-cell, 8-cell, and 24-cell, RS fails
completely on the 16-cell (0/64) and the 120-cell (0/1,440),
whereas RZ achieves Possible and Perfect results respectively.

The difference stems from the distinct nature of the two primary criteria.
RZ's max-$w$ criterion provides a \emph{global monotone descent}: at every
step, RZ moves towards the ``south pole'' of the $w$-axis, a direction that
is always well-defined and consistent throughout the traversal.
This monotonicity makes RZ robust to local dead-ends---when no right
candidate exists, the fallback (min-$w$) continues the southward descent.

RS's max-det criterion, by contrast, requires \emph{consistent spiral
curvature} relative to the fixed $\mathbf{c}_1$--$\mathbf{c}_2$ plane.
In 3D, the face-up rotation and the regular structure of Platonic solids
ensure that the spiral curvature is compatible with the right-half-space
filter at every step.
In 4D, however, the equatorial hyperplane ($w=0$) contains cell centroids
in a three-dimensional arrangement that reflects the geometry of the cell
polyhedron.
For the 120-cell, whose cells are regular dodecahedra, the 12 neighbors
of each cell form an icosahedral arrangement in the equatorial hyperplane.
The max-det criterion selects the candidate that maximizes signed
four-dimensional volume, but the icosahedral symmetry creates configurations
where the right-half-space filter and the max-det selection are mutually
incompatible: the filter admits only candidates whose determinant is
near zero, and among those, max-det provides no meaningful ordering.
The result is systematic failure.
RZ avoids this by falling back to max-$w$, which resolves the degeneracy
via a globally consistent criterion.

A further structural property of RZ on regular 4-polytopes reinforces this picture.
For any regular 4-polytope after cell-centroid-up, the cells adjacent
to $C_1$---those sharing a face with $C_1$ (the \emph{first-neighbor
shell}; for the 120-cell this coincides with band~8 in the latitude-band
partition of Section~\ref{sec:ex-120cell})---lie at exactly equal
$w$-distances from the north pole: the symmetry stabiliser of $+w$ acts
transitively on these centroids.
The max-$w$ primary criterion therefore cannot distinguish among
first-shell candidates, so the geo-score tiebreaker governs the entire
first-shell traversal.
Since $\mathbf{c}_1 = (0,0,0,w_1)$ after cell-centroid-up, the
four-dimensional determinant reduces at $k\ge 3$ to
\[
  \det(\mathbf{c}_1,\mathbf{c}_2,\mathbf{c}_k,\mathbf{c}_j)
  \;=\; -w_1 \cdot \det_{\!xyz}(\mathbf{c}_2,\mathbf{c}_k,\mathbf{c}_j),
\]
so for $k\ge 3$ the ordering within the first-neighbor shell is governed
by a three-dimensional determinant criterion applied to the centroids of
the cell polyhedron of $C_1$.
At $k=2$, where the volume form degenerates, the algorithm uses the
$xy$-area score (Algorithm~\ref{alg:peel4d}), the two-dimensional
analogue of the same azimuthal-turn criterion.
This reduction holds for all six regular 4-polytopes: the 120-cell uses
the regular dodecahedron, the 8-cell uses the cube, the 24-cell uses the
octahedron, and the 5-cell, 16-cell, and 600-cell each use the regular
tetrahedron as their local cell polyhedron.
As a consequence, for \emph{every} regular 4-polytope all first-shell
cells are exhausted before any cell at deeper $w$ is visited, regardless
of the $C_2$ choice.
For the 120-cell in particular, this explains the universal inner core
observed in Section~\ref{sec:ex-120cell}: the set of cells comprising
the first-neighbor shell is always the same regardless of $C_2$ choice,
and structural divergence into distinct spiral patterns appears only
once the algorithm leaves the first-neighbor shell (at $k=14$, i.e.,
from band~7 onward in the latitude-band partition).

The 600-cell exemplifies the structural limit of the greedy approach
(see Section~\ref{sec:ex-600cell} for the computational analysis and
worked example).
At the higher level the obstruction is that the
RZ rule, by repeatedly preferring max-$w$ neighbors, fills the
\emph{outer shell} of cells (those of locally maximal $w$) before
their unique lower-$w$ exits are reached; in the 4-regular 600-cell
neighborhood graph this leaves the local exits combinatorially blocked,
whereas in the 12-regular 120-cell graph an alternative lower-$w$
neighbor remains available at every step.
The high local connectivity of the 120-cell---rather than the global
fallback rule---is what allows RZ to clear every starting pair.

\subsection{Relation to the Companion Implementation}

The algorithm implemented here differs in three respects from the
companion paper~\cite{Yoshino2026arXiv}.
First, the right-half-space filter there uses the \emph{current} face
centroid $\mathbf{c}_k$ as reference, with the strict condition
$(\mathbf{c}_k \times \hat{z})_{xy}\cdot\mathbf{c}_{j,xy}>0$,
rather than the global fixed reference $\mathbf{c}_1$ with tolerance
$\varepsilon$ adopted here.
Second, tie-breaking among candidates with equal $z$ resolves by list
order (first occurrence in the face-index sequence) rather than by the
equivariant min-Det criterion.
Third, the RS (min-Det) rule is new to the present paper; the companion
implements only max-$z$ selection.

The local reference and list-order tie-break implicitly break equivariance:
the filter direction rotates at every step, and tie outcomes depend on the
internal ordering of faces in the data structure.
Consequently, the companion implementation reports fractional success rates
for symmetric solids---for instance, 53.3\% for the dodecahedron---whereas
equivariance theory predicts all-or-nothing rates.
The global $\mathbf{c}_1$ reference, equivariant min-Det tie-break, and
$\varepsilon$-threshold adopted here restore equivariance and are more
robust to floating-point rounding.

\section{Conclusion}

We evaluated two Apple-Peel Unfolding rules (RS, RZ) on 24 objects
across two dimensions.
The main findings are:
\begin{enumerate}
  \item \textbf{RZ dominates in 4D.}
        Under the global $\mathbf{c}_1$--$\mathbf{c}_2$ reference algorithm, RZ achieves
        Perfect on the 5-cell, 8-cell, 24-cell, and \textbf{120-cell} (100\%),
        and Possible on the 16-cell (31.3\%).
        RS is Impossible on the 16-cell and 120-cell, and matches RZ only
        on the 5-cell, 8-cell, and 24-cell.
  \item \textbf{RZ achieves 100\%} on three of thirteen Archimedean
        solids (truncated octahedron, truncated icosahedron, truncated
        cuboctahedron), and 66.7\% on the truncated icosidodecahedron;
        the snub cube and truncated tetrahedron yield partial success
        (RZ 40\% and RZ 33.3\%, respectively);
        seven of thirteen remain Impossible under both rules.
  \item \textbf{RZ outperforms RS on these four solids}:
        truncated octahedron (100\% vs.\ 41.7\%),
        truncated cuboctahedron (100\% vs.\ 0\%),
        truncated icosahedron (100\% vs.\ 0\%),
        truncated icosidodecahedron (66.7\% vs.\ 0\%).
        The sole case where RS outperforms RZ is the
        \textbf{truncated tetrahedron} (RS 66.7\% vs.\ RZ 33.3\%).
  \item The 600-cell and seven Archimedean solids are Impossible;
        impossibility correlates with structural non-uniformity or
        high combinatorial complexity.
        The 600-cell terminates at one of five fixed steps for all pairs,
        indicating icosahedral bottlenecks.
\end{enumerate}

Several concrete directions are open for future work.
First, \emph{hybrid and parametrised selection rules}: a one-parameter
family $\lambda\,z + (1-\lambda)\,\det(\mathbf{c}_1,\mathbf{c}_k,\mathbf{c}_j)$
interpolating between RZ ($\lambda{=}1$) and RS ($\lambda{=}0$) would
allow a systematic search for a crossover value at which currently
Impossible solids (such as the cuboctahedron or icosidodecahedron)
become Possible.
Second, \emph{characterizing which Apple-Peel orderings produce
geometrically valid 3D nets}: every 120-cell ordering succeeds at the
combinatorial level but every 3D realization self-intersects, so a
predictive criterion---formulated in terms of the Darboux-frame
torsion along the spiral or per-cell handedness choices at hinge
faces---would convert combinatorial success into a 3D-printability
guarantee.
Third, \emph{random convex polyhedra}: applying RS and RZ to convex
hulls of random points on $S^2$ (in 3D) or $S^3$ (in 4D) would test
how the Perfect/Possible/Impossible distribution scales with vertex
count and which combinatorial invariants (Euler characteristic of
the dual graph, average face valence, sphericity of the centroid
distribution) are predictive of unfoldability.
Fourth, a \emph{rigorous structural lemma for the 600-cell}: the
worked example in Section~\ref{sec:discussion} suggests that the 4-regular
cell-adjacency graph forces a dead-end once the outer shell of cells
is consumed; a group-theoretic argument exploiting the icosahedral
symmetry of the cell graph would upgrade the empirical impossibility
to a formal theorem.

\section*{Acknowledgements}

One of the authors (T.Y.) thanks Prof.\ Keimei Kaino, emeritus professor
of the National Institute of Technology Sendai College,
for comments on the results during the preparation of the manuscript.
This work was supported in part by the Inoue Grant of Toyo University.
This research also received funding through the Toyo University
Short-term International Visiting Professor Program in 2022,
and partial support from Chiang Mai University.

\paragraph{ORCID.}
\noindent
Takashi Yoshino: \href{https://orcid.org/0000-0003-1756-0162}{0000-0003-1756-0162};
Supanut Chaidee: \href{https://orcid.org/0000-0002-2314-1397}{0000-0002-2314-1397}.

\paragraph{Declaration of competing interest.}
The authors declare that they have no known competing financial
interests or personal relationships that could have appeared to
influence the work reported in this paper.

\section*{Code and data availability}

The Mathematica implementation of the algorithm (both 3D and 4D),
the coordinate data of the six regular 4-polytopes, and the
STL files of all valid 3D realizations of 4-polytope unfoldings
(280 nets covering the 5-, 8-, 16-, and 24-cells) are publicly
available at
\url{https://github.com/takashi-randomwalker/apple-peel-4d}.
The numerical results reported in this paper can be reproduced
by running the driver scripts included in the repository
(\texttt{run4DGlobal\_update.m} and \texttt{run4DGlobal\_120cell.m}
for the 4-polytopes, \texttt{run\_platonic\_updated.m} and
\texttt{run\_archimedean\_faceup.m} for the 3D polyhedra).


\end{document}